\documentclass{article}

\usepackage{PRIMEarxiv}

\usepackage{microtype}

\usepackage{epsfig}
\usepackage{epstopdf}
\usepackage{graphicx}
\usepackage{booktabs} 
\usepackage{caption}
\usepackage{subcaption}
\usepackage{float}

\usepackage{hyperref}

\usepackage{amsthm}
\usepackage{amssymb}
\usepackage{amsmath}
\usepackage{bm}
\usepackage{array}
\usepackage{bigstrut,rotating}

\usepackage{multirow}
\usepackage{multicol}
\usepackage{color,xcolor}

\usepackage[normalem]{ulem} 

\newtheorem*{theorem*}{Theorem}

\newtheorem{theorem}{Theorem}

\newtheorem{definition}[theorem]{Definition}

\newtheorem{example}[theorem]{Example}
\newtheorem{lemma}[theorem]{Lemma}

\newtheorem{proposition}[theorem]{Proposition}

\usepackage{algorithm}
\usepackage{algorithmic}
\renewcommand{\algorithmicrequire}{\textbf{Input:}}
\renewcommand{\algorithmicensure}{\textbf{Output:}}

\newcommand{\bR}{\mathbb{R}}

\newcommand{\cN}{\mathcal{N}}
\newcommand{\cM}{\mathcal{M}}

\newcommand{\cX}{\mathcal{X}}
\newcommand{\obj}{\operatorname{obj}}

\allowdisplaybreaks[4]

\title{Approximate Group Fairness for Clustering
}

\author{
  Bo Li \\
  Department of Computing \\
  The Hong Kong Polytechnic University \\
  \texttt{comp-bo.li@polyu.edu.hk} \\
   \And
  Lijun Li \\
  Department of Computer Science \\
  City University of Hong Kong\\
  \texttt{lijunli1211@gmail.com} \\
  \AND
  Ankang Sun \\
  Warwick Business School  \\
  University of Warwick \\
  \texttt{phd18as@mail.wbs.ac.uk} \\
  \And
  Chenhao Wang\thanks{Corresponding author.} \\
  Beijing Normal University, Zhuhai \&\\
   BNU-HKBU United International College\\
  \texttt{chenhwang@bnu.edu.cn} \\
  \And
  Yingfan Wang \\
  Department of Computer Science \\
  Duke University \\
  \texttt{yingfan.wang@duke.edu} \\
}

\begin{document}
\maketitle

\begin{abstract}
We incorporate group fairness into the algorithmic centroid clustering problem, where $k$ centers are to be located to serve $n$ agents distributed in a metric space. We refine the notion of proportional fairness proposed in [Chen et al.,  ICML 2019] as {\em core fairness}, and $k$-clustering is in the core if no coalition containing at least $n/k$ agents can strictly decrease their total distance by deviating to a new center together. Our solution concept is motivated by the situation where agents are able to coordinate and utilities are transferable.  A string of existence, hardness and approximability results is provided. Particularly,  we propose two dimensions to relax core requirements: one is on the degree of distance improvement, and the other is on the size of deviating coalition.  For both relaxations and their combination, we study the extent to which relaxed core fairness can be satisfied in metric spaces including line, tree and general metric space, and design approximation algorithms accordingly.

\end{abstract}


\section{Introduction}

Motivated by various real-world machine learning algorithm deployments where the data points are real human beings who should be treated unbiasedly, fairness is increasingly concerned.
Most traditional algorithms mainly focus on the efficiency or profit, and thus fail to ensure fairness for individual point or collection of points.
Accordingly, the past several years have seen considerable efforts in developing fair learning algorithms \cite{chierichetti2017fair,chen2019proportionally,backurs2019scalable,bera2019fair}.

Following \cite{chen2019proportionally}, we revisit the group fairness in unsupervised learning -- specifically, centroid clustering. 
A canonical clustering problem is described as: 
given a metric space $\mathcal{X}$ with distance measure $d: \mathcal{X} \times \mathcal{X} \to {\mathbb R}^+ \cup \{0\}$, a multiset $\mathcal{N} \subseteq \mathcal{X}$ of $n$ data points (in this work, each data point is an {\em agent}), a set $\mathcal{M} \subseteq \mathcal{X}$ of possible centers and a positive integer $k$, 
the task is to find a subset $Y \subseteq \mathcal{M}$ of $k$ cluster centers and assign each data point to its closest center in $Y$. 
The commonly studied objective is to make data points to be as close to their assigned centers as possible. Standard algorithms, such as $k$-means and $k$-medians, solve the clustering problem by satisfying a global criterion, where individual or group-wise happiness has not been taken into consideration.
It has been noted that the globally efficient solutions are less preferred, especially when the application scenario is about public resources allocation \cite{conitzer2017fair,fain2018fair}.
We consider the following facility location problem proposed in \cite{chen2019proportionally,michaproportionally}.

\begin{example}
\label{example:city-suburb}
Suppose the government plans to build $k=11$ identical parks to serve the residents, where every resident's cost (e.g. gasoline) is proportional to the distance between her home and the closest park.  There is a dense urban center with a population of $10,000$ and $10$ small suburbs, each of which has a population of 100.
Suppose the suburbs are close to each other (e.g. 10km) compared with the distance between them and the urban center (e.g. 500km).

\end{example}

Accordingly, by $k$-means or $k$-medians algorithms, the government will build 1 park at the urban center, and 10 parks for each small suburb.
It is not hard to see such a plan is not fair: a single park is used to serve $10,000$ people in the urban area, but each small suburb of 100 people has its own a park.
This intuition is formalized by Moulin as the {\em principle of equal entitlement} \cite{moulin2004fair}: A solution is fair if it respects the entitlements of groups of agents, namely, every subset of agents that is of sufficiently large size is entitled to choose a center for themselves.
Such group-wise fairness is also referred as {\em core}, which has been extensively studied in game theory \cite{deng1994complexity,chalkiadakis2011computational}, social choice \cite{mckelvey1986covering,feldman2006welfare} and fair division \cite{abdulkadirouglu1998random,fain2018fair}.

Chen et al. \cite{chen2019proportionally} first formally and theoretically studied group fairness in the clustering problem.
Informally, a $k$-clustering is proportional if there is no coalition of agents with size at least $n/k$ such that by deviating together to a new center each member of this coalition gets strictly better off.
We note that proportionality overlooks the situation where the agents in a coalition can facilitate internal monetary transfers between themselves if they can file a justified claim to build a park at a new location that is better for the coalition overall.

In this work, we refine proportionality by {\em core with transferable utilities} or {\em core} for short.
Formally, the cost of agent $i \in \mathcal{N}$ induced by a cluster center $y \in \mathcal{M}$ is $d(i, y)$, i.e., the distance  between $i$ and $y$; and the cost induced by a $k$-clustering $Y$ is $d(i,Y) \triangleq \min_{y \in Y} d(i, y)$, i.e., the minimum distance from $i$ to any
cluster center in $Y$.

\begin{definition}[Core]
\label{def:core}
For any $k$-clustering $Y$, a group of agents $S$ with $|S| \geq \frac{n}{k}$ is called a {\em blocking coalition}, if there is a new center $y' \in \cM \setminus Y$ such that by deviating to $y'$ together, the
total distance of $S$ can be strictly decreased, i.e.,
$\sum_{i\in S}d(i,y') < \sum_{i\in S} d(i, Y)$.
A $k$-clustering is called in the {\em core} or a {\em core clustering} if there is no blocking coalition.
\end{definition}

It is not hard to verify that core fairness is stronger than proportionality \cite{chen2019proportionally} in the sense that a core clustering must be proportional, but not vice versa.
Particularly, in Example \ref{example:city-suburb}, although a proportional clustering builds 10 parks for the city center and 1 park for all the suburbs, the one for suburbs can be arbitrarily built at any suburb's location.
However, a core clustering will select a more preferred location for this park to serve the 10 suburbs by minimizing their total distance.
Finally, it can be shown that traditional learning algorithms (such as $k$-means) can be arbitrarily bad with respect to core fairness.

\subsection{Main Contribution}

As core clusterings are not guaranteed to exist for all instances, 
in this work, we provide two relaxation dimensions to weaken the corresponding requirements. 
The first dimension is in parallel with \cite{chen2019proportionally,michaproportionally}, where a valid blocking coalition should have large distance improvement.
In the second dimension, different from their works, 
we study the relaxation when the size of a valid blocking coalition is required to be large enough. 
We formalize the two relaxations in the following definition.

\begin{definition}[Approximate Core]
For $\alpha\ge 1$ and $\beta \ge 1$, we say a $k$-clustering $Y$ is in the {\em $(\alpha,\beta)$-core} or an {\em $(\alpha,\beta)$-core clustering} if there is no $S \subseteq \mathcal{N}$ with $|S| \ge \alpha \cdot  \frac{n}{k}$ and $y'\in \mathcal M\backslash Y$ such that $\beta \cdot\sum_{i \in S}d(i,y') < \sum_{i \in S}d(i,Y)$.
\end{definition}

We investigate to what extent these two relaxations can be satisfied by finding the smallest $\alpha$ and $\beta$ such that $(\alpha,\beta)$-core is nonempty.
The relaxation on the size of blocking coalitions is regarded as $\alpha$-dimension and that on the distance improvement as $\beta$-dimension.

We consider both general metric space and special cases, including real line and discrete tree.
Line and tree are two widely studied metric spaces, 
which is partly because they extensively exist in the real world. For example, line can be used to model the situations where people want to set proper temperatures for classrooms or schedule meeting times \cite{feldman2013strategyproof},
and trees can be used to model referral or query networks \cite{kleinberg2005query,babaioff2012bitcoin}.

As argued in \cite{michaproportionally},  though objectives like truthfulness \cite{alon2010strategyproof} and social welfare maximization \cite{feldman2013strategyproof} have received significant attention, fairness is largely overlooked.

\vspace{-4mm}
\begin{table}[htbp]
    \caption{Our results for $(1,\beta)$-core.}
    \vskip 0.15in
  \centering
  {\small  \begin{tabular}{|c|c|c|c|}
    \hline
       & Line & Tree & Metric Space \bigstrut\\
    \hline
    \multirow{2}{*}{Upper Bound}    & $O(\sqrt{n})$   & $O(\sqrt{n})$ & $2\lceil\frac nk\rceil+1$ \bigstrut\\
    & (Thm \ref{thm:up1}) & (Thm \ref{thm:999}) & (Thm \ref{thm::16}) \\
    \hline
    \multirow{2}{*}{Lower Bound}     & $\Omega(\sqrt{n})$  & $\Omega(\sqrt{n})$  & $\Omega(\sqrt{n})$  \bigstrut\\
    & (Thm \ref{thm:beta:lb}) & (Thm \ref{thm:999}) & (Thm \ref{thm:beta:lb}) \\
    \hline 
    \end{tabular}}
    \label{tab:1}
\end{table} 

\vskip 0.1in

\paragraph{$(1,\beta)$-Core.}
We first, in Section \ref{sec:beta}, study the relaxation in $\beta$-dimension, where $\alpha$ is fixed to be 1, i.e., $(1,\beta)$-core. Our results are summarized in Table \ref{tab:1}. 
Different to the study of proportionality in \cite{chen2019proportionally}, where a constant approximation can be guaranteed for any metric space, we show that for core fairness, the existence of $(1,o(\sqrt{n}))$-core clustering is not guaranteed, even in a real line. 
On the other hand, when the metric space is a real line or a tree, we present efficient algorithms that always output a $(1, O(\sqrt{n}))$-core clustering, 
and thus we get the optimal approximation algorithm by relaxing the distance requirement solely.

With respect to general metric space, we show that a greedy algorithm ensures $O(\frac nk)$-approximation.
Beyond the study for arbitrary number $k$, 

when $k \ge \frac{n}{2}$, we show that for any metric space, there is a polynomial time algorithm which returns a $(1,2)$-core clustering, whereas 
determining the existence of $(1,2-\epsilon)$-core clustering is NP-complete. 

\vspace{-4mm}
\begin{table}[htbp]
    \caption{Our results for $(\alpha,1)$-core. }
    \vskip 0.15in
  \centering
 {\small   \begin{tabular}{|c|c|c|c|}
    \hline
       & Line & Tree & Metric Space \bigstrut\\
    \hline
    \multirow{2}{*}{Upper Bound}   & $2$  & $2$ & $k$  \bigstrut\\
    & (Thm \ref{thm:up3line}) & (Thm \ref{thm:up3tree}) & (Thm \ref{thm:clique}) \\
    \hline
    \multirow{2}{*}{Lower Bound}     & $2$  & $2$  & $\min\{k,\max\{\frac k2,\frac n4\}\}$ \bigstrut\\
    & (Thm \ref{thm:alpha:line:lb}) & (Thm \ref{thm:up3tree}) & (Thm \ref{thm:clique}) \\
    \hline
    \end{tabular} }
     \label{tab:2}
\end{table}

\vskip 0.1in

\paragraph{$(\alpha,1)$-Core.}
Next, in Section \ref{sec:4}, we study the relaxation in $\alpha$-dimension, where $\beta$ is fixed to be 1, i.e., $(\alpha,1)$-core. Our results are summarized in Table \ref{tab:2}.
Different to the relaxation in $\beta$-dimension, we prove a $(2,1)$-core clustering is guaranteed to exist in line and tree spaces. We complement this result with a line instance where $(2-\epsilon,1)$-core is empty for any $\epsilon > 0 $. Thus our algorithms are optimal.

For general metric space, we observe that a trivial upper-bound for $\alpha$ is $k$, which can be guaranteed by 
placing the $k$ centers such that the total distance of all agents is minimized. We also complement this observation with a lower-bound instance where $(\alpha, 1)$-core is empty for any $\alpha \le \min\{k,\max\{\frac k2,\frac n4\}\}$, and thus our algorithm is tight up to a constant factor. 
Finally, we end this section by proving that determining the existence of $(\alpha,1)$-core clustering for any constant $\alpha\ge 1$ in general metric space is NP-complete.

\paragraph{$(\alpha,\beta)$-Core}
In Section \ref{sec:alpha-beta}, we integrate previous results and study the case when both dimensions can be relaxed. 
Intuitively, sacrificing the approximation ratio in one dimension should be able to improve that in the other.
We prove this intuition affirmatively by quantifying the tradeoff between the two relaxation dimensions. Specifically,
\begin{itemize}
    \item for line or tree space and any $\alpha\in(1,2]$, $(\alpha, \frac{1}{\alpha-1})$-core is always non-empty (Thm \ref{thm:ab});
    \item for general metric space and $\alpha>1$, $(\alpha,\max\{4, \frac{2}{\alpha -1} + 3\})$-core is always non-empty (Thm \ref{thm:ab_general_space}). 
\end{itemize}
 We want to highlight the significance of the above two results, especially for the general metric space, which is regarded as the major theoretical contribution of the current work.
The results in Sections \ref{sec:beta} and \ref{sec:4} imply that, in the general metric space, if $\alpha=1$ is not relaxed, the best possible approximation ratio for $\beta$ is $\Theta(\sqrt{n})$;
on the other hand, if $\beta =1$ is not relaxed, the best possible approximation ratio for $\alpha$ is $\max\{\frac k2,\frac n4\}$.
However, our results in this section show that if we sacrifice a small constant on the approximation for one dimension, we can guarantee constant approximation for the other dimension.
For example, by relaxing $\alpha$ to $2$, $(2, 5)$-core is always non-empty,
and by relaxing $\beta$ to $4$, $(3, 4)$-core is always non-empty.

\paragraph{Experiments}
Finally, in Section \ref{sec:experiment}, we conduct experiments to examine the performance of our algorithms.
We note that our algorithms have good theoretical guarantees in the worst case, but they may not find the fairest clustering for every instance. 
Accordingly, we first propose a two-stage algorithm to refine the clusters and then use synthetic and real-world data sets to show how much it outperforms classic ones regarding core fairness.
Actually, the second stage of our algorithm provides us an interface to balance fairness and social efficiency.
As shown by the experiments, our solution does not sacrifice much efficiency.

\subsection{Other Related Works}

{\em Cooperative Game Theory.} 
Core stability is widely studied in cooperative game theory with transferable utilities, where a collective profit is to be distributed among a set of agents \cite{chalkiadakis2011computational}.
Informally, a core payoff vector ensures no coalition wants to deviate from the grand coalition. 
Particularly, in Myerson games
\cite{DBLP:journals/mor/Myerson77,DBLP:conf/aaai/MeirZER13,  DBLP:conf/wine/Elkind14}, the agents
are located on a graph and the coalitions are required to be connected components in this graph.
The differences between our model and theirs include the following perspectives. 
First, instead of a distribution of a (divisible and homogeneous) profit, an outcome in our model is the locations for $k$ facilities. 
Second, the coalitions are required to be sufficiently large in our model instead of being connected. 
In cooperative games, when the core is empty, {\em cost of stability} \cite{DBLP:conf/sagt/BachrachEMPZRR09} or {\em least core} \cite{DBLP:journals/mor/MaschlerPS79} are studied, where additive relaxations are imposed to offer subsidies to agents if they stay in the grand coalition or penalize deviating coalitions. 
In our model, $\alpha$-dimension relaxation works on the size of blocking coalitions and $\beta$-dimension works on the distance improvement, and both relaxations are defined in the multiplicative way.

{\em Fairness Study in Machine Learning.} The necessity of incorporating fairness into machine learning algorithms has been well recognized in the recent decade. Various fairness concepts have been proposed based on different principles and to adapt to different machine learning tasks. For example, in another seminal work, Chierichetti et al. \cite{chierichetti2017fair} also studied fairness issue under clustering context but defined fairness as preserving equal representation for each protected class in every cluster based on disparate impact in US law. 
To satisfy their fairness requirement, Chierichetti et al. \cite{chierichetti2017fair} designed a two-step algorithm which first decomposes data points into fairlets and then runs classical clustering algorithms on those fairlets.
On one hand, the algorithm is improved by a number of subsequent works in the sense of approximation ratios \cite{harb2020kfc,bercea2019cost} and running time \cite{huang2019coresets,schmidt2018fair,backurs2019scalable}. 
On the other hand, Bera et al.\cite{bera2019fair}, Rosner et al. \cite{rosner2018privacy} and Braverman et al. \cite{braverman2019coresets} extended the setting in \cite{chierichetti2017fair} to allow overlap in protected groups, consider multiple protected features or address fairness towards remote data points.
Under the context of classification, fairness is captured by envy-freeness in \cite{DBLP:conf/nips/ZafarVGGW17,DBLP:conf/icml/UstunLP19,DBLP:conf/nips/BalcanDNP19}.
We refer readers to the survey by \cite{mehrabi2019survey} for a detailed discussion on fairness study in various machine learning contexts. 

{\em Fair Resource Allocation.} Group fairness is recently considered in the resource allocation field when resources are allocated among groups of agents.
Based on how groups are formed, the models can be classified into two categories.
The first one is when agents are partitioned into fixed groups and the fairness is only concerned with the pre-existing groups \cite{segal2019democratic,benabbou2019fairness,kyropoulou2020almost}.
The second one is to consider the fairness of arbitrarily formed groups of agents \cite{aziz2019almost,berliant1992fair,conitzer2019group,hossain2020designing}.
Our problem falls under the umbrella of the second category.
However, the items in our work are public and fairness is defined for collective utilities. 
Fair clustering problem is also related to public resource allocation, where resources can be shared by agents. 
For public resources, one popular research agenda is to impose (combinatorial) constraints on the allocations \cite{conitzer2017fair, cheng2019group,fain2018fair,li2020fair}.
However, in our setting, all centers will be built without constraints.

\section{Preliminaries}

Recall that $\cN \subseteq \cX $ is a multiset of $n$ agents in a metric space $(\mathcal X,d)$, where for any $\{x_1, x_2, x_3\} \subseteq \mathcal{X}$,  $d(x_1, x_2) + d(x_2, x_3) \ge d(x_1, x_3)$.
Note that we allow repetitions in $\cN$ which means multiple agents can be located at the same position. We refer to agents and their locations interchangeably, when no confusion arises.
$\mathcal M\subseteq \mathcal X$ is the set of  feasible locations for cluster centers.  
Our task is to find $Y\in \cM^k$ to place the $k$ centers such that the clustering is in the (approximate) core.
We first present an example where an exact core clustering does not exist, and illustrate the two relaxation  dimensions so that an approximate core clustering exists.

\begin{example}
Consider a complete graph $K_4=(V,E)$ with 4 vertices and the distance between any two vertices is~1. Let $\mathcal X=\mathcal M=V$. Suppose that {at} each vertex lies an agent (i.e., $\mathcal N=V$), and we want to locate $k=2$ centers to cluster these $n=4$ agents. 
First, we note that the (exact) core is empty for this instance, because for any $2$-clustering $Y\in \mathcal M^2$, the remaining $n/k=2$ agents (say $u,v$) in $V\backslash Y$ can form a blocking coalition and deviate to a new center $v\in V\backslash Y$ such that $d(u,v)+d(v,v)=1<2=d(u,Y)+d(v,Y)$. 
Second, for all $\beta\ge 2$, every solution $Y\in \mathcal M^2$ is a $(1,\beta)$-core clustering, because for any group $S\subseteq V$ with $|S|\ge 2$ and any possible deviating center $v \in V$, we have $\beta\sum_{i\in S}d(i,v)\ge2\ge \sum_{i\in S}d(i,Y)$. 
Finally, for all $\alpha>1$, every solution $Y\in \mathcal M^2$ is an $(\alpha,1)$-core clustering, because for any group $S\subseteq V$ with $|S|\ge \alpha\cdot\frac nk >2$ and any possible deviating center $v \in V$, we have $\sum_{i\in S}d(i,v)\ge2\ge \sum_{i\in S}d(i,Y)$.
\end{example}

Motivated by the above example, our task is to compute the smallest $\beta$ or $\alpha$ such that $(1,\beta)$-core or $(\alpha,1)$-core is non-empty.
Besides the general metric space, we are also interested in two special spaces, namely, real line $\bR$ and graph spaces such as tree.

\paragraph{Line.} $\mathcal X = \cM =\bR$ and the distance $d$ is Euclidean, i.e., the agents and centers can be at anywhere in the real line.

\paragraph{Graph Space.} Let $G=(V,E)$ be an undirected tree graph. The edges in $E$ may have specified length. 
In the graph space induced by $G$, $\mathcal X=\mathcal M=V$, and the distance $d$ between two vertices is the length of the shortest path. A tree space is induced by a tree graph.
Note that line is a continuous space where every point is feasible for cluster centers, while the graph space is discrete and centers can only be placed on its vertex set.

The following Lemma \ref{lem:nor} enables us to only focus on the blocking coalitions with minimum possible size. 

\begin{lemma}\label{lem:nor}
Given $\alpha,\beta\ge1$ and a solution $Y\in [\mathcal M]^k$, if there is a group $S$ of size $|S|>\lceil\frac{\alpha n}{k}\rceil$ such that $\beta\sum_{i\in S}d(i,y')<\sum_{i\in S}d(i,Y)$ for some $y'\in\mathcal M$, then there is a group $S'\subseteq S$ of size  $|S'|=\lceil\frac{\alpha n}{k}\rceil$ such that $\beta\sum_{i\in S'}d(i,y')<\sum_{i\in S'}d(i,Y)$.
\end{lemma}

\vspace{-3mm}\begin{proof}
    If $S$ and $y'\in\mathcal M$ satisfy $|S|>\lceil\frac{\alpha n}{k}\rceil$ and
    $\frac{\sum_{i\in S}d(i,Y)}{\sum_{i\in S}d(i,y')}>\beta$,
    then there must exist an agent $w\in S$ subject to $$\frac{d(w,Y)}{d(w,y')}\le \frac{\sum_{i\in S}d(i,Y)}{\sum_{i\in S}d(i,y')}.$$
    Let $S'=S\backslash\{w\}$. Then we have
    $$\beta< \frac{\sum_{i\in S}d(i,Y)}{\sum_{i\in S}d(i,y')}\le\frac{\sum_{i\in S'}d(i,Y)}{\sum_{i\in S'}d(i,y')}.$$
    Repeating this process of removing one agent until $|S'|=\lceil\frac{\alpha n}{k}\rceil$, we establish the proof.
\end{proof}

{When $k=1$ or $n-1$, the solution that minimizes the total distance to all agents is in the core.} 
When $k\ge \frac{n}{2}$ and the space is a connected graph $G(V,E)$ with $V=\cX=\cN=\cM$, by contrast with the result in \cite{michaproportionally}  where a proportional $k$-clustering always exists and can be computed efficiently, 
in our setting a core clustering is not guaranteed. We state a stronger result as follows.

\begin{proposition}\label{thm:non}
    When {$k=1$ or} $k\ge n-1$, {the core is always non-empty.} 
    When $\frac{n}{2}\le k\le n-2$, for any $0<\epsilon\le 1$, the existence of a $(1,2-\epsilon)$-core clustering is not guaranteed.
\end{proposition}

\vspace{-4mm}\begin{proof}
    For the first argument, let
    \[
    Y^* = \arg\min_{Y\in \mathcal M^{k}}\sum_{i\in N}d(i,Y), \text{ for $k=1$ or $n-1$}.
    \]
    If $k=1$, by definition it is obvious that $Y^*$ is in the core as the ground agent set $N$ is the unique possible deviating coalition.
    For $k=n-1$, a deviating coalition should contain 2 agents.
    Suppose for contradiction that $i$ and $j$ can decrease their cost by deviating to $y'\notin Y^*$. Then we can construct a new clustering $Y'$: (1) $y'\in Y'$ and thus $d(i, Y^*) + d(j,Y^*) > d(i, Y') + d(j,Y')$;
    (2) each of the other $n-2$ centers is at the optimal position for one of the remaining $n-2$ agents in $N \setminus \{i,j\}$ and thus $d(l, Y')\le d(l, Y^*)$ for $ l\in N \setminus \{i,j\}$. 
    Therefore,
    \[
    \sum_{i\in N}d(i,Y^*) > \sum_{i\in N}d(i,Y'),
    \]
    which is a contradiction with the definition of $Y^*$.

    For the second argument, we consider the graph space induced by a complete graph $G=(V,E)$ with  $n$ vertices, {where at each vertex lies an agent, and $V=\cX=\cN=\cM$.}
    
    The distance between any two vertices is 1. When $\frac{n}{2}\le k\le n-2$, the minimum size of a possible blocking coalition is $\lceil\frac{n}{k}\rceil=2$. For every $k$-clustering $Y$, there must exist two agents without center at their locations. Then they form a blocking coalition (w.r.t. $(1,2-\epsilon)$-core) with a deviating center on one of them: their total distance to the deviating center is 1, while their total distance to $Y$ is 2. Hence, a $(1,2-\epsilon)$-core clustering  does not exist.
\end{proof}


\section{$(1, \beta)$-Core}
\label{sec:beta}

In this section, we study the relaxation on distance improvement,  
and show to what extent a $(1, \beta)$-core is non-empty.

\subsection{Line}\label{subsec:line1}

First, a $(1,o(\sqrt{n}))$-core clustering is not
guaranteed to exist.
\begin{theorem}\label{thm:beta:lb}
There is a line instance such that the $(1,o(\sqrt{n}))$-core is empty.
\end{theorem}
\begin{proof}
Consider an instance building in the real line and a set of integer points $V=\{1,2,\ldots,k+1\}$ in this line, each of which accommodates $k$ agents. A $k$-clustering is required to serve all $n=k(k+1)$ agents by $k$ centers.

Suppose for contradiction that $Y$ is a $(1,o(\sqrt{n}))$-core clustering. Since there are $k+1$ agents' locations and k centers to be located, there must be a point $j\in V$ so that no center is located in the interval $(j-\frac12,j+\frac12)$, i.e., $(j-\frac12,j+\frac12)\cap Y=\emptyset$.
 Assume w.l.o.g. that $j\neq k+1$ by symmetry. Consider a group $S$ consisting of $k$ agents located at $j$ and one agent located at $j+1$. Its size is $|S|=k+1=\frac nk$, which entitles itself to choose a center. Because every agent at $j$ has a distance at least $\frac12$ from its location to solution $Y$, the total distance of this group is
 $\sum_{i\in S}d(i,Y)\ge\frac12\cdot k$.
 However, if they deviate to a new center $y'=j$, their total distance changes into
 $\sum_{i\in S}d(i,y')=1$.
Then we have
$\frac{\sum_{i\in S}d(i,Y)}{\sum_{i\in S}d(i,y')}\ge\frac k2>o(\sqrt{n})$,
which is a contradiction. Hence, the  $(1,o(\sqrt{n}))$-core is empty.

\end{proof}

\begin{algorithm}[H]
	\caption{\hspace{-2pt}~{\bf $ALG_l(\lambda)$ for Line.} 
	}
	\label{alg:line}
	\begin{algorithmic}[1]
	\REQUIRE 
	
	Agents $\mathbf x=\{x_1,\ldots,x_n\}$, number $ k \in \mathbb{N}^{+}$ 
	\ENSURE $k$-clustering $Y=\{y_1,\ldots,y_k\}$
	\STATE Rename the agents such that $x_1\le\cdots\le x_n$.
	\FOR{$i=1,2,\ldots,k-1$}
	\STATE Locate a center at $y_i=x_{\lambda i}$.
	\ENDFOR
	\STATE Let $r=\min\{\lambda k,n\}$, and locate a center at $y_k=x_{r}$.
	\end{algorithmic}
\end{algorithm}

Next, we present our algorithm $ALG_l$, as shown in Algorithm~\ref{alg:line}, which matches the lower-bound in Theorem \ref{thm:beta:lb}.
Roughly, $ALG_l$ has a tuneable parameter $\lambda$, and guarantees that the number of agents between any two contiguous centers to be no more than $\lambda-1$.
By selecting the optimal $\lambda$ depending on $k$, we can obtain the tight approximation.

\begin{theorem}\label{thm:up1}
For any line instance, a $(1,O(\sqrt n))$-core clustering can be found in linear time.
Specifically, $ALG_l(\lceil\frac{n}{k}\rceil)$ gives a $(1,\lceil\frac{n}{k}\rceil-1)$-core clustering if $k=\Omega(\sqrt n)$, and  $ALG_l(\lceil\frac{n}{k+1}\rceil)$ gives a $(1,k)$-core clustering if $k=o(\sqrt n)$.
\end{theorem}

\vspace{-3mm}\begin{proof}
    Let $\mathbf x=\{x_1,\ldots,x_n\}$ be the locations of agents.  When $k=\Omega(\sqrt n)$ and $k=o(\sqrt n)$, we define $\lambda=\lceil\frac{n}{k}\rceil$ and $\lambda=\lceil\frac{n}{k+1}\rceil$ respectively, and implement $ALG_l(\lambda)$. 
    The output is a $k$-clustering $Y=\{y_1,\ldots,y_k\}$. Now we prove this solution is a $(1,\lceil\frac{n}{k}\rceil-1)$-core clustering when $k=\Omega(\sqrt n)$, and a $(1,k)$-core clustering when $k=o(\sqrt n)$. Therefore, it is always a $(1,O(\sqrt n))$-core clustering.

    Suppose for contradiction that there is a blocking coalition $S\subseteq \cN$ of agents with $|S|\ge \frac{n}{k}$ and a deviating center $y'\in \bR\backslash Y$ satisfying
    $r:=\frac{\sum_{i\in S}d(i,Y)}{\sum_{i\in S}d(i,y')}> \lceil\frac{n}{k}\rceil-1$ (resp. $r > k$) when $k=\Omega(\sqrt{n})$  (resp. $k=o(\sqrt{n})$). 
    Set two virtual points $y_0=-\infty$ and $y_{k+1}=+\infty$. Assume w.l.o.g. $y'\in (y_j,y_{j+1})$ for some $j=0,\ldots,k$ and $d(y_j,y')\le d(y_{j+1},y')$. Let $(N_1,N_2,N_3)$ be a partition of $S$ with $N_1=\{i\in S|x_i\le y_j\}$, $N_2=\{i\in S| y_j<x_i<y_{j+1}\}$, and $N_3=\{i\in S|x_i\ge y_{j+1}\}$. Note that the algorithm guarantees $|N_2|\le \lambda-1$.  Then we have
    \begin{align}
        \sum_{i\in S}d(i,Y) &\le \sum_{i\in N_1}d(i,y_j) +\sum_{i\in N_2}(d(i,y')+d(y',y_j)) \nonumber\\
        &+\sum_{i\in N_3}d(i,y_{j+1}),  \label{eq:11}
    \end{align}
    \begin{align}
        \sum_{i\in S}d(i,y') &= \sum_{i\in N_1}(d(i,y_j)+d(y_j,y'))+\sum_{i\in N_2}d(i,y') \nonumber\\
        &+\sum_{i\in N_3}(d(i,y_{j+1})+d(y_{j+1},y')).\label{eq:22}
    \end{align}

    Combining Equations (\ref{eq:11}) and (\ref{eq:22}), it follows that 

$$r\le\frac{\sum_{i\in N_2}d(y',y_j)}{\sum_{i\in N_1}d(y_j,y')+\sum_{i\in N_3}d(y_{j+1},y')}   \le\frac{|N_2|}{ |N_1 \cup N_3|},
$$

    where the second inequality is due to the assumption $d(y_j,y')\le d(y_{j+1},y')$.

    When $k=\Omega(\sqrt n)$ and $\lambda=\lceil\frac{n}{k}\rceil$, we have  $|N_2|\le \lceil\frac{n}{k}\rceil-1$, and $|N_1\cup N_3|=|S|-|N_2|\ge 1$. It indicates that $r\le \lceil\frac{n}{k}\rceil-1$ which is a contradiction. So $Y$ is a $(1,\lceil\frac{n}{k}\rceil-1)$-core clustering.

     When $k=o(\sqrt n)$ and $\lambda=\lceil\frac{n}{k+1}\rceil$, we have $|N_2|\le \lceil\frac{n}{k+1}\rceil-1$, and $|N_1\cup N_3|=|S|-|N_2|\ge \lceil\frac{n}{k}\rceil-\lceil\frac{n}{k+1}\rceil+1$. By a simple calculation we have $r\le k$ as a contradiction. So $Y$ is a $(1,k)$-core clustering.
\end{proof}


\subsection{Tree}

Both the lower- and upper-bound results for line space can be extended to trees. 
For the proportionality fairness, Micha and Shah \cite{michaproportionally} proposed an algorithm ``Proportionally Fair Clustering for Trees (PFCT)'' as follows, which always returns a proportional solution for trees. A rooted tree $(G,r)$ is obtained by rooting the tree at an arbitrary node $r$. Let $\textnormal{level}(x)$ denote the height of node $x$ relative to the root $r$ (with $\textnormal{level}(r)=1$), and $\textnormal{ST}(x)$ denote the subtree rooted at node $x$ (i.e. the set of nodes $v$ with $\textnormal{level}(v) \le \textnormal{level}(x)$ and the unique path from $v$ to $r$ passes by $x$). Let $|ST(x)|$ be the number of agents contained in the subtree $ST(x)$. PFCT traverses all the nodes from the highest level ones (i.e. the leaves), locates a center on the node whose subtree contains at least $\lceil \frac{n}{k} \rceil$ agents, and then deletes this subtree. At the end, agents are assigned to the closest center in the output. We adapt PFCT into the following $ALG_t(\lambda)$ with a tuneable parameter $\lambda$, {which controls the threshold value for locating a center. When $\lambda=\lceil \frac{n}{k} \rceil$, $ALG_t(\lambda)$ is equivalent to PFCT.}

\begin{algorithm}[H]
	\caption{\hspace{-2pt}{ \bf $ALG_t(\lambda)$ for Tree.}}
	\label{alg:tree}
	\begin{algorithmic}[1]
	\REQUIRE A tree $G=(V,E)$, $n$ agents and integer $k$
	\ENSURE $k$-clustering set $Y$
	\STATE Let $r$ the root of tree $G$ and $d$ be the height. 
	\STATE $Y\leftarrow \emptyset$
	\STATE $ G^d \leftarrow G$
	\FOR{$l=d$ to 1}
	\STATE $G^{l-1} \leftarrow G^l$
	\FOR{$\text { every } x \in V \text { with level }(x)=\ell$  and $|\operatorname{ST}(x)| \geq \lambda$}
	\IF{$|Y|<k$}
	\STATE $Y \leftarrow Y \cup\{x\}$ and $G^{\ell-1} \leftarrow G^{\ell-1} \setminus \mathrm{ST}(x)$
	\ENDIF
	\ENDFOR
	\ENDFOR
	\end{algorithmic}
\end{algorithm}

Note that it always has $|Y|\le k$, and if $|Y|<k$, we can build $k-|Y|$ more centers arbitrarily. When $\lambda=\lceil\frac{n}{k}\rceil$ or $\lceil\frac{n}{k+1}\rceil$, we  observe that, after removing all centers in $Y$ from the tree, the number of agents in every component 
is at most $\lambda-1$.
It can be observed that $ALG_t(\lambda)$ is an extension of $ALG_l(\lambda)$ in the tree, and we have the following theorem. 

\begin{theorem}\label{thm:999}
For any instance in the tree, we can find a $(1,O(\sqrt{n}))$-core clustering efficiently. In particular,   when $k=\Omega(\sqrt n)$, $ALG_t(\lceil\frac{n}{k}\rceil)$ returns a $(1,\lceil\frac{n}{k}\rceil-1)$-core clustering, and  when $k=o(\sqrt n)$,  $ALG_t(\lceil\frac{n}{k+1}\rceil)$ returns a $(1,k)$-core clustering.
Moreover, a $(1,o(\sqrt{n}))$-core clustering is not guaranteed to exist.
\end{theorem}

\vspace{-3mm}\begin{proof}
    Define $\lambda=\lceil\frac{n}{k}\rceil$ if $k=\Omega(\sqrt{n})$, and $\lambda=\lceil\frac{n}{k+1}\rceil$ otherwise. The output of the algorithm $ALG_t(\lambda)$ is location profile $Y=(y_1,\ldots,y_k)$.
    Suppose for contradiction that there is a blocking coalition $S\subseteq \cN$ of agents with $|S|\ge \lceil\frac{n}{k}\rceil$ and a deviating center $y'\in V\backslash Y$. Let $d(u,v)$ denote the distance between nodes  $u$ and $v$, i.e., the length of the unique path from $u$ to $v$.
    {Define a partition $(N_1,N_2)$ of group $S$ (as Figure \ref{fig:Partition} shows): $N_2$ consists of the agents in $S$ who are in the component containing $y'$, where the components are obtained by removing all centers in $Y$ from the tree; $N_1$ consists of the remaining group members.}
 Then we have
    \begin{equation}\label{eq:ss}
        \sum_{i\in S}d(i,Y) \le \sum_{i\in N_1}d(i,Y) +\sum_{i\in N_2}(d(i,y')+d(y',Y)),
    \end{equation}
    \begin{align}\label{eq:tt}
        \sum_{i\in S}d(i,y') &\ge \sum_{i\in N_1}(d(i,Y)+d(Y,y'))+\sum_{i\in N_2}d(i,y').
    \end{align}
    
    \begin{figure}[h]
    \centering
    \includegraphics[width=4.2cm]{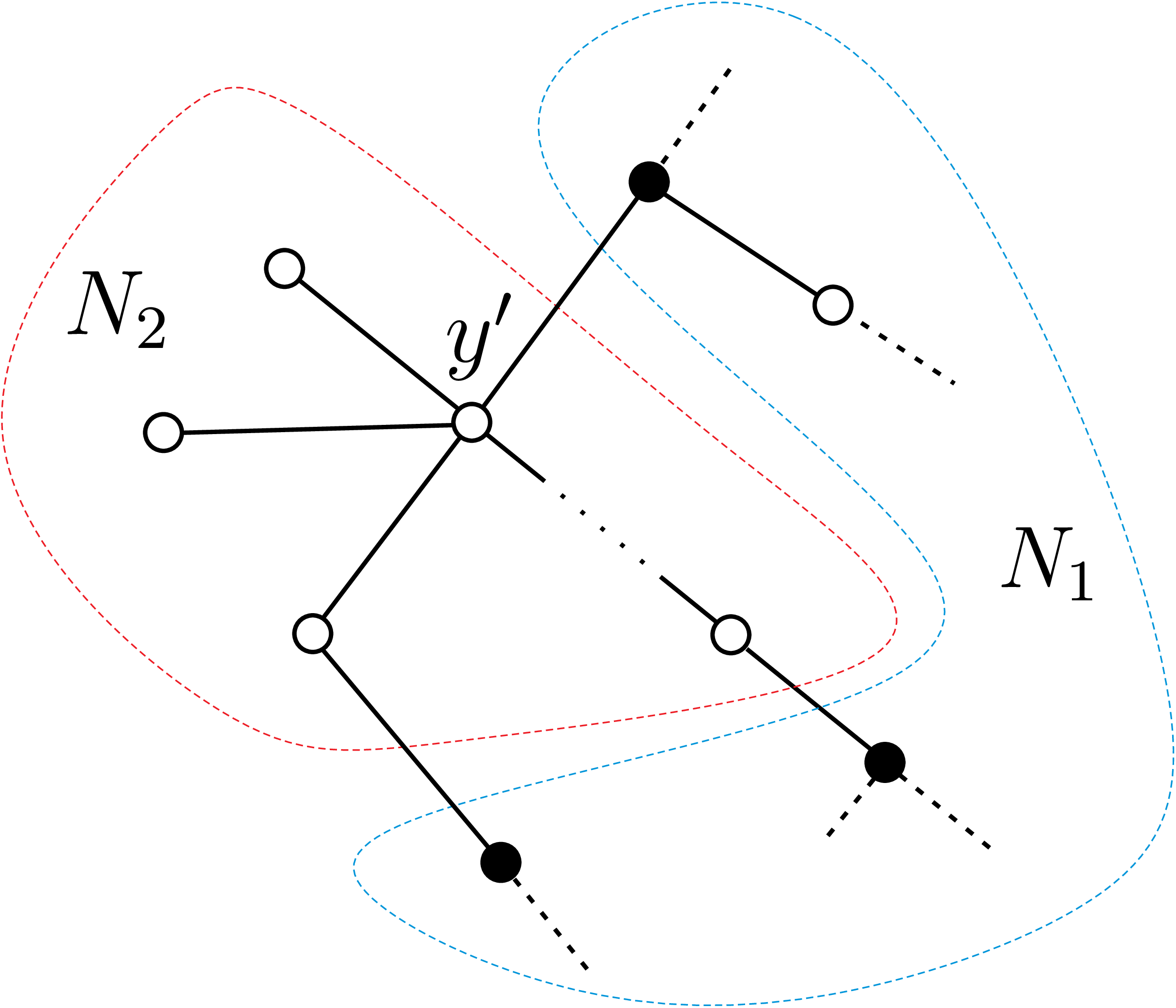}

    \caption{\label{fig:Partition}\small The partition $(N_1,N_2)$ of group $S$, where solid points indicate centers in $Y$.}
    \end{figure}

   Because $S$ is a blocking coalition, we have $r:=\frac{\sum_{i\in S}d(i,Y)}{\sum_{i\in S}d(i,y')}>1$. It follows from (\ref{eq:ss}) and (\ref{eq:tt}) that
\begin{equation}\label{eq:r}
    r\le\frac{\sum_{i\in N_2}d(y',Y)}{\sum_{i\in N_1}d(y',Y)}
        =\frac{|N_2|}{|N_1|}.
\end{equation}

   When $k=\Omega(\sqrt{n})$ and $\lambda=\lceil\frac{n}{k}\rceil$, the algorithm guarantees that $|N_2|\le \lceil\frac{n}{k}\rceil-1$, and thus $|N_1|=|S|-|N_2|\ge 1$, implying that $r\le \lceil\frac{n}{k}\rceil-1$ and $S$ cannot be a blocking coalition w.r.t. $(1,\lceil\frac{n}{k}\rceil-1)$-core.  Therefore, $Y$ is a $(1,\lceil\frac{n}{k}\rceil-1)$-core clustering.

   When $k=o(\sqrt{n})$ and $\lambda=\lceil\frac{n}{k+1}\rceil$, the algorithm guarantees that $|N_2|\le \lceil\frac{n}{k+1}\rceil-1$, and thus $|N_1|=|S|-|N_2|\ge \lceil\frac{n}{k}\rceil-\lceil\frac{n}{k+1}\rceil+1$. A simple computation shows that $r\le k$, indicating that $S$ cannot be a blocking coalition.  So $Y$ is a $(1,k)$-core clustering.
   
   For the lower bound, we note that the instance constructed in the line in the proof of Theorem \ref{thm:beta:lb} can be adapted to the graph space induced by a path graph, where $V=\{1,2,\ldots,k\}$ is the vertex set of the path, and each vertex accommodates $k+1$ agents. Using the similar analysis, it can be shown that the $(1,o(\sqrt{n}))$-core for this path is empty. 
\end{proof}



\subsection{General Metric Space}\label{sec:32}
For the general metric space, we show that a simple greedy algorithm, $ALG_g$, as described in Algorithm \ref{alg:social_optimal}, has the desired theoretical guarantee. 
For each $x\in \mathcal{X}$, we use $B(x, \delta) = \{i \in \cN \mid d(i,x) \le \delta\}$ to denote the set of agents in the ball with center $x$ and radius $\delta$. $ALG_g$ continuously grows the radius of each ball centered on each possible center, with the same speed. When a ball is large enough to contain at least 
$\lceil \frac{n}{k} \rceil$ points, we open a cluster center at that ball center. Actually, the underlying idea of expanding balls of points has been utilized for $k$-median problems \cite{jain1999primal}, and 
\cite{chen2019proportionally} and \cite{michaproportionally} have proved that $ALG_g$ also has good theoretical performance regarding proportional fairness. {We note that $ALG_g$ may output less than $k$ centers, and if so, the remaining centers can be selected arbitrarily. } In Section \ref{sec:experiment}, we will show how to refine $ALG_g$ beyond the worst case.

\begin{algorithm}[H]
	\caption{\hspace{-2pt}{ \bf $ALG_g$ for General Metric Space. }}
	\label{alg:social_optimal}
	\begin{algorithmic}[1]
	\REQUIRE Metric space $(\mathcal{X},d)$, agents $\cN \subseteq \cX$, possible locations $\cM \subseteq \cX$, and $k \in \mathbb{N}^{+}$
	\ENSURE $k$-clustering $Y$
	\STATE $\delta \leftarrow 0$ ; $Y \leftarrow \emptyset$ ; $N \leftarrow \mathcal{N}$.
	\WHILE{$N\neq \emptyset$}
	\STATE Smoothly increase $\delta$
	\WHILE{$\exists x \in Y \text { s.t. }|B(x, \delta) \cap N| \geq 1$}
	\STATE $N \leftarrow N \backslash B(x, \delta)$
	\ENDWHILE
	\WHILE{$\exists x \in\mathcal{M}\backslash Y$ s.t. {$|B(x, \delta) \cap N|\ge \lceil \frac{n}{k} \rceil$}} 
	\STATE $Y \leftarrow Y \cup\{x\}$ and $N \leftarrow N \backslash B(x, \delta)$
	\ENDWHILE
	\ENDWHILE
	\end{algorithmic}
\end{algorithm}

\begin{theorem}\label{thm::16}
	For any instance in metric space $(\mathcal{X}, d )$, algorithm $ALG_g$ always outputs a $(1,2\lceil\frac{n}{k}\rceil + 1)$-core clustering.
\end{theorem}
\vspace{-3mm}\begin{proof}
	Let $Y$ be the k-clustering returned by $ALG_g$. By Lemma \ref{lem:nor}, it suffices to prove that, for any set of agents $S\subseteq \mathcal{N}$ with $|S| = \lceil\frac{n}{k} \rceil$, $\sum_{i \in S} d(i,Y) \leq (2\lceil\frac{n}{k} \rceil+1)\sum_{i \in S} d (i ,y' )$ holds for any point $y'\in \mathcal M \setminus Y$.
	Let $y^{*}\in Y$ be the center closest to $y'$, i.e., $y ^ {*} \in \arg\min_{y \in Y} d (y, y' )$. Since for any $i \in S$, $d(i, Y ) \leq d(i, y^*)$, we have
	\begin{equation}\label{eq:1}
		\sum _ { i \in S} d(i,Y) \leq \sum_{ i \in S} d(i ,y ^*) \leq \sum_{ i \in S} d (i ,y') + \sum_{ i \in S} d (y^* ,y').
	\end{equation}
	We discuss two cases.
	
	\emph{Case 1:} $\max_{i\in S} d (i,y') \geq \frac{1}{2} d (y^*, y' )$. If so, then we have
	\[
	\sum_{ i \in S} d (i,y') \geq \max _ { i \in S} d (i,y') \geq \frac{1}{2}d(y^*, y').
	\]
	 Combining with (\ref{eq:1}), it follows that

    \begin{align*}
        \frac{\sum_{i\in S}d(i,Y)}{\sum_{i\in S}d(i, y' )} &\le 1+\frac{\sum _ { i \in S} d (y^*, y' )}{\sum_{i\in S}d(i,y')} \\
        & \leq 1+\frac{\sum_{i\in S}d(y^*,y')}{d(y^*, y')/2}=1+ 2\lceil\frac{n}{k} \rceil.
    \end{align*}
	
	\emph{Case 2:} $\max_{i\in S}d(i,y')<\frac{1}{2}d(y^*,y')$.  Let {$\delta^* = \max _ { i \in S} d (i,y')$}, and accordingly, {$S\subseteq B(y',\delta^*)$}.  If there exists a center $y'' \in Y$ such that {$B(y'' , \delta^*) \cap B(y' ,\delta^*) \neq \emptyset$}, then we have {$d(y'' , y' ) \leq 2\delta^* <d(y^*,y')$}, which however, contradicts to the definition of $y^*$. So {when the algorithm operating radius $\delta$ as $\delta^*$}, all balls with center in $Y$ have an empty intersection with ball {$B(y', \delta^*)$}, that is, {$B(y', \delta^*) \cap B(y, \delta^*) =\emptyset, \forall y \in Y$}. Since $|S| = \lceil \frac{n}{k} \rceil$ and {$S\subseteq B(y',\delta^*)$}, point $y'$ must be selected as a center by the algorithm, contradicting to $y' \notin Y$. Therefore, this case would never occur. 
\end{proof}


When the number $k$ of cluster centers is large, we are able to get better approximations.  Recall Proposition \ref{thm:non} that when $k\ge n-1$, the core is always non-empty.
We complement this result by the following theorem.

\begin{theorem}\label{thm:prim}
    For any metric space with $\frac{n}{2}\le k\le n-2$, there is an algorithm that computes a $(1,2)$-core clustering in $O(n^2\log n)$ time.
\end{theorem}

\vspace{-3mm}\begin{proof}
    Construct an $n$-vertex complete graph by setting the vertices to be the locations of $n$ agents, and the weight of each edge to be the distance of the two endpoints. Consider the following algorithm: run Prim's algorithm to find a minimum spanning tree (MST) over the complete graph, and then find a vertex cover of size $k$ for the MST, which can be done efficiently by selecting every second vertex. The algorithm locates $k$ center on the locations in this vertex cover. 
    Note that every two adjacent agents in the MST have at least one center located on themselves. Now we show this solution (denoted by $Y$) is a $(1,2)$-core clustering. We refer to agents and vertices interchangeably.    The time complexity comes from Prim's algorithm.

    First, by Lemma \ref{lem:nor}, we only need to consider groups consisting of $2$ agents, denoted by $\{u,v\}$. Second, if one group member has a center on his/her location, then it cannot be a blocking coalition, because the group's total distance to any new center is at least $d(u,v)$ by the metric, while the total distance to $Y$ is at most $d(u,v)$. So we only need to consider group $\{u,v\}$ without center on $u$'s nor $v$'s location. By the definition of vertex cover, edge $(u,v)$ is not in the MST. Recall that Prim's algorithm operates by building the tree one vertex at a time, from an arbitrary starting vertex (say $r$), at each step adding the cheapest possible connection from the tree to another vertex. Assume $u$ is added earlier than $v$, and the vertex set at the time of adding $u$ is $W$. Assume $v$ is added via an edge $(y,v)$. By Prim's algorithm, it must be $d(u,v)\ge d(y,v)$, otherwise $v$ is added by edge $(u,v)$, a contradiction. We discuss the following two cases. 

    (a) If the unique $r$-$v$ path in the MST passes through $u$, denote by $(r,\ldots,u,z,\ldots,v)$. It is easy to see that $d(u,z)\le d(u,v)$, otherwise edge $(u,v)$ is added into the MST. Thus we have $2d(u,v)\ge d(y,v)+d(u,z)$. Note that both agents $y$ and $z$ have facilities on their locations. The total distance of group $\{u,v\}$ to $Y$ is at most $d(y,v)+d(u,z)$, and that to any deviating center is at least $d(u,v)$. It indicates that $\{u,v\}$  cannot form a blocking coalition w.r.t. $\beta=2$. Hence, $Y$ is a $(1,2)$-core clustering.

(b) If the unique $r$-$v$ path in the MST does not pass through $u$, denote by $(r,\ldots,w,z,\ldots,v)$, and assume $w\in W$, $z\notin W$. It is easy to see that $d(u,W)\le d(z,W)$, otherwise $z$ is added into $W$, a contradiction. {Because $z$ is added earlier than $v$, we have $d(z,W)\le d(w,z)\le d(u,v)$.} 
Thus we have $2d(u,v)\ge d(y,v)+d(z,W)\ge d(y,v)+d(u,W)$. Because both the agent $y$ and the agent incident to $u$ in $W$ have centers on their locations, the total distance of group $\{u,v\}$ to $Y$ is at most $d(y,v)+d(u,W)$, while that to any deviating center is at least $d(u,v)$. So $\{u,v\}$  cannot form a blocking coalition w.r.t. $\beta=2$, and $Y$ is a $(1,2)$-core clustering.
\end{proof}

The $(1,2)$-core clustering provided above is best possible, because by Proposition \ref{thm:non}, a $(1,2-\epsilon)$-core clustering may not exist. Finally, we end this section by proving that deciding the existence of a $(1,2-\epsilon)$-core clustering is hard.

\begin{theorem}\label{thm:np}
For any $0<\epsilon\le 1$, the problem of determining the existence of a $(1,2-\epsilon)$-core clustering is NP-complete, {even if $k\ge \frac n2$ and it is in a graph space induced by $G=(V,E)$ with $|V|=n$.}
\end{theorem}

\vspace{-3mm}\begin{proof}
    We give a reduction from the minimum vertex cover problem (MVC). Let  $(G'=(V',E'),k')$ with $|V'|=n'$ be an arbitrary instance of the MVC, which asks whether a vertex cover of $G'$ with size $k'$ exists. We construct an instance $I$ of our clustering problem with $n=4n'$ agents and $k=2n'+k'$ centers as follows. In instance $I$, {the space is induced by a graph $G=(V,E)$ with} $n'+1$ connected components, one being graph $G'$ and other $n'$ components being triangles $K_3$. We assume w.l.o.g. that every two components are connected by an edge with infinite weight, and all other edges have unit weight. Define the distance between any two vertices to be the length of the shortest path between them. At each vertex of $G$ lies an agent, so there are $4n'$ agents in total. We show that the answer of $(G',k')$ is ``yes" if and only if $I$ has a $(1,2-\epsilon)$-core clustering. By Lemma \ref{lem:nor}, we only need to consider a group of agents of size $\lceil\frac{n}{k}\rceil=2$ as a possible blocking coalition. 

     If $G'$ has a vertex cover $W\subseteq V'$ of size $k'$, then we have a $k$-clustering of $G$ (denoted by $Y\subseteq V$) which locates $k'$ centers on vertices in $W$, and locates two centers on vertices of each triangle. Note that no possible blocking coalition can have two agents in different components.  Clearly, any two agents in a triangle cannot form a blocking coalition, as their distance to $Y$ is at most 1.  Thus it suffices to consider a group $S\subseteq V'$ consisting of two agents in $G'$. If no member of $S$ is in the vertex cover $W$, then there is no edge between the two members, and the total distance of $S$ to any new center is at least $2$. Since the total distance to $Y$ is $2$, group $S$ cannot form a blocking coalition.  If there is some member in $W$, then the total distance to $Y$ is at most $1$, implying that no deviation can occur. Hence, there is no blocking coalition, indicating that our solution is a core clustering (and thus a $(1,2-\epsilon)$-core clustering).

     If $I$ has a $(1,2-\epsilon)$-core clustering $Y$, it must locate at least two centers on each triangle (otherwise two agents in a triangle form a blocking coalition), and thus locate at most $k'$ centers on $G'$. For any two adjacent agents $u,v\in V'$ in $G'$, if neither of them has a center located on themselves, they can deviate to a new center $u$ such that $$d(u,u)+d(v,u)=1<\frac{2}{(2-\epsilon)}\le \frac{d(u,Y)+d(v,Y)}{2-\epsilon},$$
     implying that they can form a blocking coalition. Therefore, there is at least one center between any two adjacent agents in $G'$, and thus
      the center locations in $G'$ is a vertex cover.
\end{proof}



\section{$(\alpha, 1)$-Core}\label{sec:4}


Next, we study to what extent core fairness can be achieved when only the requirement of blocking coalition's size is relaxed, i.e., finding the smallest $\alpha$ such that $(\alpha, 1)$-core is non-empty.


\subsection{Line}\label{subsec:line2}

\begin{theorem}
\label{thm:alpha:line:lb}
In line space, for any $\epsilon > 0$, there exists an instance such that the $(2-\epsilon,1)$-core is empty. 
\end{theorem}
\vspace{-3mm}\begin{proof}
Consider an instance in the real line with $n=C(2C-1)$ agents and $k=2C-1$ centers to be opened for some integer $C\ge \frac{3}{\epsilon}$. Let  $K$ be a sufficiently large number, say $K = C^n$. The agents are partitioned into $C$ parts $(N_1, \cdots, N_C)$ such that for any $j \in [C]$, $N_j = \{i_1, i_2, \cdots, i_{2C-1}\}$ contains $2C-1$ agents where $\{i_1, \cdots, i_{C-1}\}$ are located at $jK$, agent $i_C$ is located at $jK + 1$ and $\{i_{C+1}, \cdots, i_{2C-1}\}$ are located at $jK+2$.

As we only open $2C-1$ centers, for any $(2-\epsilon,1)$-core clustering, there exists $N_j$ such that all agents in $N_j$ are served by a single center. Since the agents in $\{i_1, \cdots, i_{C-1}\}$ and $\{i_{C+1}, \cdots, i_{2C-1}\}$ are symmetric with respect to $i_C$, assume w.l.o.g. the center is placed at $x \ge jK+1$. Note that $2C-3\ge (2-\epsilon)\lceil\frac nk\rceil$. We show that $S=\{i_1,\ldots,i_{2C-3}\}$ forms a blocking coalition w.r.t. $(2-\epsilon,1)$-core by deviating to a new center $i_1$:
\begin{align*}
    \sum_{i\in S}d(i,i_1) &=1+2(C-3) < 2C-4 \\
    & =\sum_{i\in S}d(i,i_C)  \le \sum_{i\in S}d(i,x).
\end{align*}
Hence, there does not exist a $(2-\epsilon,1)$-core clustering.
\end{proof}

Fortunately, by selecting a proper parameter $\lambda$,  Algorithm $ALG_l$ guarantees the tight approximation ratio.

\begin{theorem}\label{thm:up3line}
For any line instance, $ALG_l(\lceil\frac{n}{k}\rceil)$  returns a $(2,1)$-core clustering. 
\end{theorem}

\vspace{-3mm}\begin{proof}
   For any instance in the line, we show that $ALG_l(\lceil\frac nk\rceil)$ outputs a $(2,1)$-core clustering. Let $\mathbf x=\{x_1,\ldots,x_n\}$ be the locations of agents. Recall that the algorithm $ALG_l(\lceil\frac{n}{k}\rceil)$ locates a center at $y_i=x_{\lceil\frac{n}{k}\rceil\cdot i}$ for $i=1,\ldots,k-1$, and $y_k=x_n$. The output is centers $Y=\{y_1,\ldots,y_k\}$. Now we prove that $Y$ is a $(2,1)$-core clustering.

   Consider an arbitrary group $S\subseteq \cN$ of agents with $|S|\ge \frac{2n}{k}$, and an arbitrary point $y'\in \bR\backslash Y$.   Set two virtual points $y_0=-\infty$ and $y_{k+1}=+\infty$. Assume w.l.o.g. $y'\in (y_j,y_{j+1})$ for some $j=0,\ldots,k$.

   Let $N_1=\{i\in S|y_j<x_i<y_{j+1}\}$ be the set of members who are located in the interval $(y_j,y_{j+1})$, and $N_2=S\backslash N_1$ be the set of members who are outside. By the algorithm, we have $|N_1|\le \lceil\frac{n}{k}\rceil-1$, and $|N_2|=|S|-|N_1|\ge |N_1|$.  We construct pairs \[
   \{(u_1,v_1),(u_2,v_2),\ldots,(u_{|N_1|},v_{|N_1|})\}
   \]
   between $N_1$ and $N_2$, so that  $u_p\in N_1,v_p\in N_2$ and $u_p\neq u_q,v_p\neq v_q$, for any different $p,q\le|N_1|$. Then it is easy to see that for any pair of agents, their total distance to the potential deviating center $y'$ cannot be smaller than that to the solution $Y$, i.e.,
   $$d(u_p,y')+d(v_p,y')\ge d(u_p,Y)+d(v_p,Y)$$
   for any $p\le|N_1|$. Further, for those agents in $N_2$ but not in any pair (if any), they also have a smaller distance from solution $Y$ than that from $y'$. Therefore, $S$ cannot be a blocking coalition deviating to $y'$. By the arbitrariness of $S$ and $y'$, $Y$ is a $(2,1)$-core clustering.
\end{proof}

\subsection{Tree}

We continue to consider the tree spaces. 

\begin{theorem}\label{thm:up3tree}
For any tree instance, we can find a $(2,1)$-core clustering efficiently. For any $\epsilon > 0$, a $(2-\epsilon,1)$-core clustering is not guaranteed to exist in a tree.
\end{theorem}

\vspace{-3mm}\begin{proof}
For any instance in the tree, we show that $ALG_t(\lceil\frac nk\rceil)$ outputs a $(2,1)$-core clustering.
    Let  $Y$ be the solution output by $ALG_t(\lceil\frac nk\rceil)$.
    Suppose for contradiction that there is a blocking coalition $S\subseteq \cN$ with $|S|\ge \frac{2n}{k}$ and a deviating center $y'$.
    {Define a partition $(N_1,N_2)$ of group $S$: $N_2$ consists of the agents in $S$ who are in the component containing $y'$, where the components are obtained by removing all centers in $Y$ from the tree; $N_1$ consists of the remaining group members.}

Recall that $ALG_t(\lceil\frac nk\rceil)$ guarantees
\[
|N_2|\le \lceil\frac{n}{k}\rceil-1
\]
and thus
\[
|N_1|\ge |S|-(\lceil\frac{n}{k}\rceil-1)\ge \lceil\frac{n}{k}\rceil+1.
\]
By Equation  (\ref{eq:r}), we have
 $$\frac{\sum_{i\in S}d(i,Y)}{\sum_{i\in S}d(i,y')}\le \frac{|N_2|}{|N_1|}\le \frac{\lceil\frac{n}{k}\rceil-1}{\lceil\frac{n}{k}\rceil+1}<1.$$
This contradicts the definition of a blocking coalition.

  For the lower bound, we note that the instance constructed in the line in the proof of Theorem \ref{thm:alpha:line:lb} can be adapted to the graph space induced by a path graph. Using the similar analysis, it can be shown that the $(2-\epsilon,1)$-core for this path is empty. 
\end{proof}

\subsection{General Metric Space}\label{sec:42}

By the definition of $(\alpha,1)$-core clustering, a $(k,1)$-core clustering always exists, because any potential blocking coalition must contain all agents and any solution containing the optimal single center is a  $(k,1)$-core clustering. Next, we show that this trivial upper-bound is asymptotically tight.

\begin{theorem}\label{thm:clique}
{The $(k,1)$-core is always non-empty. Further,} the existence of an $(\alpha,1)$-core clustering, for any $\alpha\le \min\{k,\max\{\frac k2,\frac n4\}\}$, is not guaranteed. 
\end{theorem}

\vspace{-3mm}\begin{proof}
Consider the graph space induced by a complete graph $K_n=(V,E)$ with $n=2k$ vertices, where at each vertex lies an agent, and the distance between any two vertices is 1. We show that, for $\alpha\le\frac k2=\frac n4$, the $(\alpha,1)$-core is empty. Let $Y\subseteq V$ with $|Y|=k$ be an arbitrary $k$-clustering. Consider the group $V\backslash Y$ of agents with size $|V\backslash Y|=k\ge \alpha\cdot\frac nk$.
The total distance of this group to $Y$ is $k$, as each member has a distance 1 to $Y$. However, they may deviate to a new center $v\in V\backslash Y$, such that the total distance becomes $\sum_{i\in V\backslash Y}d(i,v)=k-1<k$. So they form a blocking coalition. This completes the proof.
\end{proof}

Finally, we give a hardness result for the existence of  an $(\alpha,1)$-core clustering.

\begin{theorem}\label{thm:np99}
For any given constant $\alpha\ge 1$, the problem of determining the existence of an $(\alpha,1)$-core clustering is NP-complete. 
\end{theorem}
\vspace{-3mm}\begin{proof}
 We give a reduction from the minimum vertex cover problem. Let $(G'=(V',E'),k')$  be an arbitrary instance of the MVC with $|V'|=n'\ge 2k'$ vertices, which asks whether a vertex cover of $G'$ with size $k'$ exists. {Define an integer $x:=n'\lceil\alpha\rceil + 2\lceil2\alpha\rceil -2-2k'$.} 


We construct an instance $I$ of our problem on an edge-weighted connected graph $G=(V,E)$ as follows, with $|V|=n$ and $V=\cX=\cM$. $G$ is constructed based on $G'$:
connect each vertex $v\in V'$ with $\lceil\alpha\rceil-1$ new vertices by unit-weighted edges (denote this graph by $G''$), and connect an arbitrary vertex of $G''$ with an $x$-vertex complete graph $K_x$ by a large-weighted edge, as shown in Figure~\ref{fig:13}. 
There are $n:=|V|=\lceil\alpha\rceil n'+x$ vertices, at each of which lies an agent, and $k:=k'+x-\lceil2\alpha\rceil+1$ centers are to be located.
All edges in $G$ have unit weight except the one connecting $G''$ and $K_x$ with a large weight, which guarantees that no center can serve agents in $G''$ and $K_x$ simultaneously.

\begin{figure}[h]
    \centering
    \includegraphics[width=8cm]{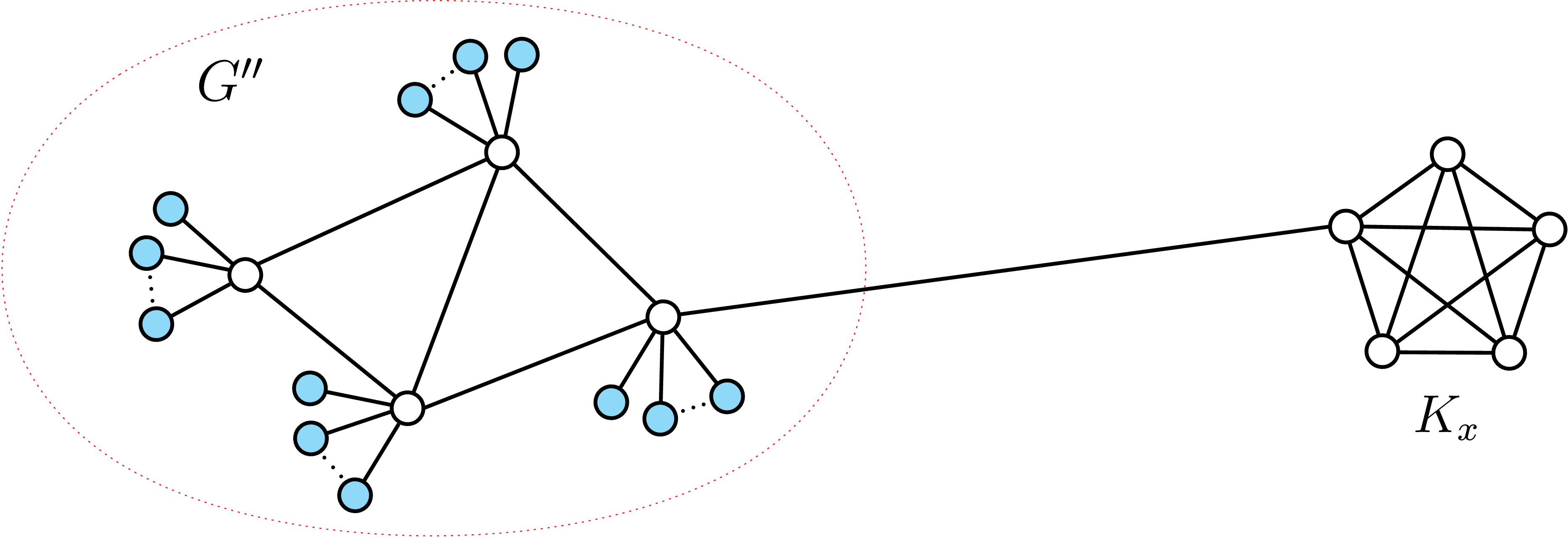}
    \caption{\small An illustration for $G$, where $G''$ is obtained by adding  $\lceil\alpha\rceil-1$ new vertices (denoted by blue nodes) for each vertex in $G'$, and $G''$ and $K_x$ are connected by a large-weighted edge.}
    \label{fig:13}
\end{figure}

Note that $\frac{\alpha\cdot n}{k}=2\alpha$.  By  Lemma \ref{lem:nor}, it suffices to consider the groups of size $\lceil2\alpha\rceil$.  
We show that, $G'$ admits a vertex cover $D\subseteq V'$ of size $k'$, if and only if $G$ has an $(\alpha,1)$-core clustering. 

If $G'$ has a vertex cover $D$ of size $k'$, we construct a $k$-clustering $Y$ on $G$, which locates $k'$ centers on $D$, and locates $k-k'=x-\lceil2\alpha\rceil+1$ centers on different vertices of $K_x$. Note that no possible blocking coalition can have agents in both $K_x$ and $G''$. For any group $S$ of size $\lceil2\alpha\rceil$ in $K_x$, at least one member has a center, and thus the total distance to $Y$ is at most $\lceil2\alpha\rceil-1$. However, the total distance to any deviating center is at least  $\lceil2\alpha\rceil-1$, implying that $S$ cannot be a blocking coalition. On the other hand, for any group $S$ of size $\lceil2\alpha\rceil$ in $G''$, if it is a blocking coalition, there must exist a deviating center $y'\in V'\backslash D$ (because a new vertex is a deviating center only if its adjacent vertex in $V'$ is a deviating center). Note that by the definition of vertex cover, all neighbors of $y'$ in $V'$ must have a center in $Y$.   Denote the set of agents who are located at the $\lceil\alpha\rceil-1$ new vertices adjacent to $y'$ or $y'$ itself by $S'\subset S$. By deviating to $y'$, each agent in $S'$ decreases his/her cost by 1, while each agent in $S\backslash S'$ increases his/her cost by 1. Since $|S'|\le \lceil\alpha\rceil$, the total cost of $S$ cannot be improved by deviating.
So $Y$ is in the $(\alpha,1)$-core.

 If $G'$ has no vertex cover of size $k'$, suppose for contradiction that $I$ admits an $(\alpha,1)$-core clustering $Y$. First, it must locate at least  $k-k'=x-\lceil2\alpha\rceil+1$ centers on $K_x$, as otherwise those $\lceil2\alpha\rceil$ agents without center can form a blocking coalition. So it locates at most $k'$ centers on $G''$. Because $G'$ has no vertex cover of size $k'$, there must be two vertices $u,v\in V'$ so that there is no center among $u,v$ and the new vertices adjacent to them. 
 These $2+2(\lceil\alpha\rceil-1)=2\lceil\alpha\rceil\ge \lceil2\alpha\rceil$ vertices can form a blocking coalition by deviating to $u$, giving a contradiction.
\end{proof}

\section{$(\alpha, \beta)$-Core}
\label{sec:alpha-beta}

Finally, we study the approximate fairness by relaxing both dimensions simultaneously.
Recalling the results in Sections \ref{subsec:line1} and \ref{subsec:line2}, if $\alpha = 1$, the best possible approximation of $\beta$ is $\Theta(\sqrt n)$; however, if $\alpha$ is relaxed to $2$, we are able to get the optimum in the $\beta$-dimension, i.e., $\beta = 1$.
The following theorem shows the exact tradeoff between the approximations in both dimensions for line and tree spaces.

\begin{theorem}\label{thm:ab}
For any $\alpha>1$, every line or tree instance has a non-empty $(\alpha,\beta)$-core with $\beta=\max\left\{1,\frac{1}{\alpha-1}\right\}$.
\end{theorem}
\vspace{-3mm}\begin{proof}
We first prove for lines. Let $\mathbf x=\{x_1,\ldots,x_n\}$ be the locations of agents. Consider the algorithm $ALG_l(\lceil\frac{n}{k}\rceil)$:  locate a center at $y_i=x_{\lceil\frac{n}{k}\rceil\cdot i}$ for $i=1,\ldots,k-1$, and locate a center at $y_k=x_n$. The output is $Y=\{y_1,\ldots,y_k\}$. We show that $Y$ is in $(\alpha,\beta)$-core for any $\alpha>1$ and $\beta=\max\{1, \frac{1}{\alpha-1}\}$. 

    Suppose for contradiction that there is a blocking coalition $S\subseteq \cN$ with $|S|\ge \alpha\cdot\frac nk$ and a deviating center $y'\in \bR\backslash Y$.  Set two virtual points $y_0=-\infty$ and $y_{k+1}=+\infty$. Assume w.l.o.g. $y'\in (y_j,y_{j+1})$ for some $j=0,\ldots,k$ and $d(y_j,y')\le d(y_{j+1},y')$.
    Let $(N_1,N_2,N_3)$ be a partition of $S$ with $N_1=\{i\in S|x_i\le y_j\}$, $N_2=\{i\in S| y_j<x_i<y_{j+1}\}$, and $N_3=\{i\in S|x_i\ge y_{j+1}\}$. Note that the algorithm guarantees $|N_2|\le \lceil\frac{n}{k}\rceil-1$.

     Because $S$ is a blocking coalition, we have
   $r:=\frac{\sum_{i\in S}d(i,Y)}{\sum_{i\in S}d(i,y')}> \beta\ge 1$.
   However, as the proof of Theorem \ref{thm:up1}, we have
$$r\le\frac{|N_2|}{|N_1 \cup N_3|}\le \frac{\lceil\frac{n}{k}\rceil-1}{\frac{\alpha n}{k}-\lceil\frac{n}{k}\rceil+1}\le \beta$$
    where the second equation is because $|N_2|\le \lceil\frac{n}{k}\rceil-1$,
and $|N_1\cup N_3|= |S|-|N_2|\ge \frac{\alpha n}{k}-\lceil\frac{n}{k}\rceil+1$.
This is a contradiction.

For any instance in the tree, we show that $ALG_t(\lceil\frac nk\rceil)$ outputs a $(\alpha,\beta)$-core clustering for any $\alpha>1$.
Let  $Y$ be the solution output by $ALG_t(\lceil\frac nk\rceil)$.
Suppose for contradiction that there is a blocking coalition $S\subseteq \cN$ with $|S|\ge \frac{\alpha n}{k}$ and a deviating center $y'$.
{Define a partition $(N_1,N_2)$ of group $S$: $N_2$ consists of the agents in $S$ who are in the component containing $y'$, where the components are obtained by removing all centers in $Y$ from the tree; $N_1$ consists of the remaining group members.}

Recall that $ALG_t(\lceil\frac nk\rceil)$ guarantees $|N_2|\le \lceil\frac{n}{k}\rceil-1$ and thus
$|N_1|\ge |S|-(\lceil\frac{n}{k}\rceil-1)\ge \frac{\alpha n}{k}-\lceil\frac{n}{k}\rceil+1$.
By Equation  (\ref{eq:r}), we have
 $$\frac{\sum_{i\in S}d(i,Y)}{\sum_{i\in S}d(i,y')}\le \frac{|N_2|}{|N_1|}\le \frac{\lceil\frac{n}{k}\rceil-1}{\frac{\alpha n}{k}-\lceil\frac{n}{k}\rceil+1}\le \beta.$$
This contradicts the definition of a blocking coalition.
\end{proof}

Next, we extend the above theorem to general metric space.
Again our results in Sections \ref{sec:32} and \ref{sec:42} show that if $\alpha$-dimension is not relaxed, the best approximation in $\beta$-dimension is $\Theta(\sqrt{n})$;
and if $\beta$-dimension is not relaxed, the best approximation in $\alpha$-dimension is $\max\{\frac k2,\frac n4\}$.
As we will see in the following theorem, however, if we sacrifice a small constant for one dimension, we can guarantee constant approximation for the other as well.

\begin{theorem}\label{thm:ab_general_space}
For any $\alpha>1$, any instance in a metric space has a non-empty $(\alpha,\beta)$-core with
$\beta = \max\left\{4, \frac{2}{\alpha -1} + 3\right\}$.
\end{theorem}
\vspace{-3mm}\begin{proof}
    Let $Y$ be the $k$-clustering returned by $ALG_g$. By Lemma \ref{lem:nor}, it suffices to prove that, for any set of agents $S\subseteq \mathcal{N}$ with $|S| = \lceil\frac{\alpha n}{k} \rceil$,
    \[
    \sum_{i \in S} d(i,Y) \leq \max \left\{ 4, \frac{2}{\alpha-1}+3 \right\} \cdot \sum_{i \in S} d (i ,y' )
    \]
    holds for any point $y'\in \mathcal M \setminus Y$.
    Suppose for contradiction that there is a blocking coalition $S\subseteq \cN$ with $|S| = \lceil \frac{\alpha n}{k} \rceil$ and a deviating center $y'\in \mathcal{M}\backslash Y$.

    Let $y^{*}\in Y$ be the cluster center closest to $y'$, i.e., $y ^ {*} \in \arg\min_{y \in Y} d (y, y' )$. Note that for any $i \in S$, $d(i, Y ) \leq d(i, y^*)$. 
     Because $S$ is a blocking coalition, we have
   \[
   r:=\frac{\sum_{i\in S}d(i,Y)}{\sum_{i\in S}d(i,y')}> \beta = \max \left\{ 4, \frac{2}{\alpha-1}+3 \right\}.
   \]

   Consider an open ball $X$ centered at $y'$ with radius
   $R := \frac{d(y',y^*)}{2}$, $X=\{i \in S~|~ d(i,y')< R\}$.
   Note that $|X| \le\lceil \frac{n}{k} \rceil - 1$, otherwise by $ALG_g$, $y'$ should be selected as a center.
    
   Define
   $$
   r_1 := \frac{ \sum_{i \in X} d(i,Y)}{\sum_{i \in S} d(i,y')}~\text{~and~}~
   r_2 := \frac{ \sum_{i \in S \backslash X} d(i,Y)}{\sum_{i \in S} d(i,y')}.
   $$
   Then $r=r_1 + r_2$.
   For $r_1$, if $r_1 \geq 1$, we have,

   \begin{equation*}
       \begin{split}
           r_1 \leq& \frac{ \sum_{i \in X} d(i,y') + \sum_{i \in X} d(y',y^*)}{ \sum_{i \in X} d(i,y') + \sum_{i \in S \backslash X} d(i,y')} \le   \frac{\sum_{i \in X} d(y',y^*)}{\sum_{i \in S \backslash X} d(i,y')} \\
           \leq &  \frac{|X|\cdot 2R}{|S \backslash X| \cdot R} 
           \leq \frac{ 2 \left( \lceil \frac{n}{k} \rceil -1 \right)}{\lceil \frac{\alpha n}{k} \rceil - \lceil \frac{n}{k} \rceil +1}
           \le \frac{2}{\alpha -1}.
       \end{split}
   \end{equation*}
  
    Combining with the other case $r_1 < 1$, we have $r_1 \le \max\left\{ 1, \frac{2}{\alpha-1} \right\}$.
  Then, we derive an upper-bound for $r_2$.
   \begin{equation*}
       \begin{split}
       r_2 
       & \leq \frac{\sum\limits_{ i \in S \setminus X} (d(i,y^{\prime}) + d(y^{\prime}, y^{*}))}{\sum\limits_{ i \in S \setminus X} d(i,y^{\prime})} \leq 1 + \frac{|S \setminus X| \cdot 2R}{ |S \setminus X| \cdot R} = 3,
       \end{split}
   \end{equation*}
   where the last inequality is due to $d(y^{\prime},y^*) = 2R$ and $d(i,y^{\prime}) \geq R$ for any $i \in S \setminus X$. Combing the upper-bounds of $r_1,r_2$, we have $
   \beta < r = r_1 + r_2  \leq \max\{4, \frac{2}{\alpha + 1} + 3\} = \beta$,
   which is a  contradiction.
\end{proof}

\section{Experiments}
\label{sec:experiment}

\subsection{A Refined Algorithm $ALG_g^+(\obj)$}\label{sec:000}
Before examining the performance of our algorithm,
we note that though $ALG_g$ has good guarantee in the worst-case, it may not produce the fairest clustering in every instance. 
For example, in Figure \ref{fig:y equals x}, we randomly partition each of 1,000 nodes into 3 Gaussian-distributed sets with probability 0.2, 0.3 and 0.5 (from left to right). We want to build $k=10$ centers to serve the nodes; however, $ALG_g$ only returns 4 clusters whose centers are shown by the red stars. 
Obviously, the extra 6 centers can significantly improve its performance, regarding both fairness and efficiency.

\begin{figure}[h]
     \centering
     \begin{subfigure}[b]{0.33\linewidth}
         \centering
         \includegraphics[width=2.27in]{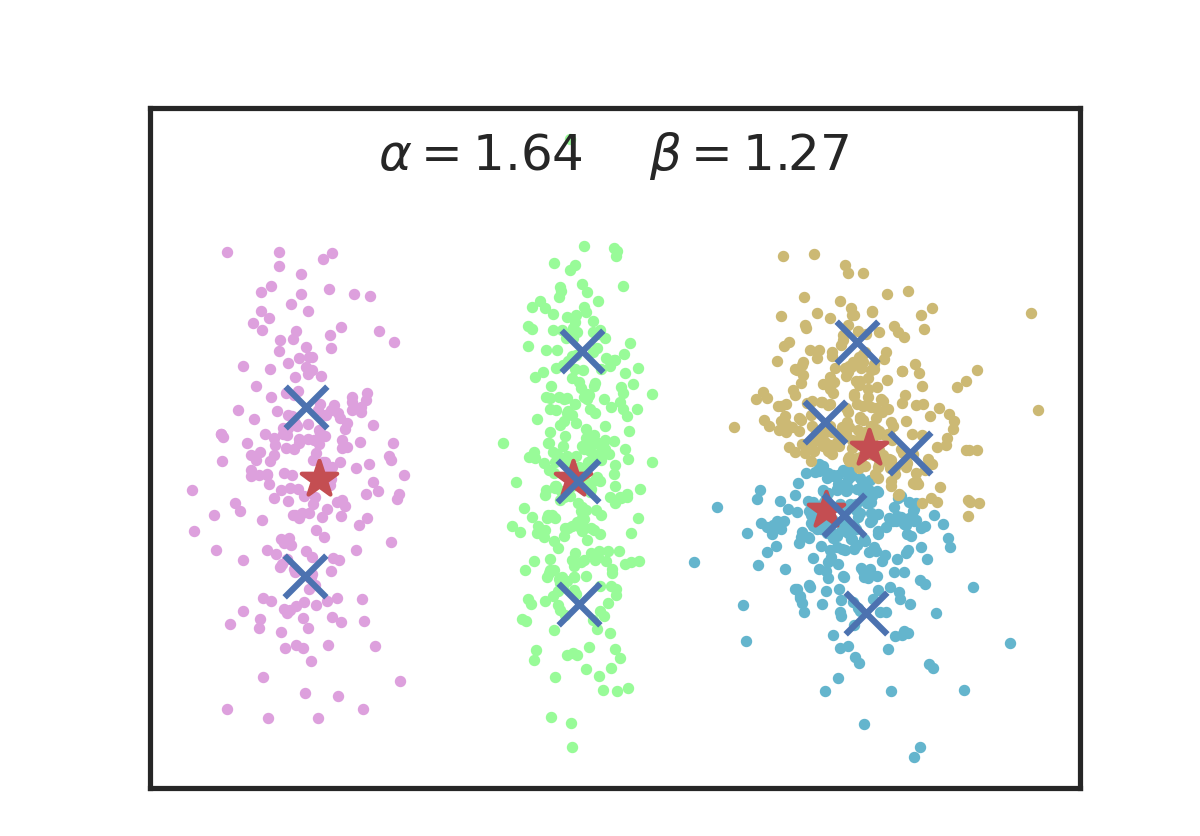}
         \caption{$ALG_g^+(k\text{-means})$}
         \label{fig:y equals x}
     \end{subfigure}
     \begin{subfigure}[b]{0.33\linewidth}
         \centering
         \includegraphics[width=2.24in]{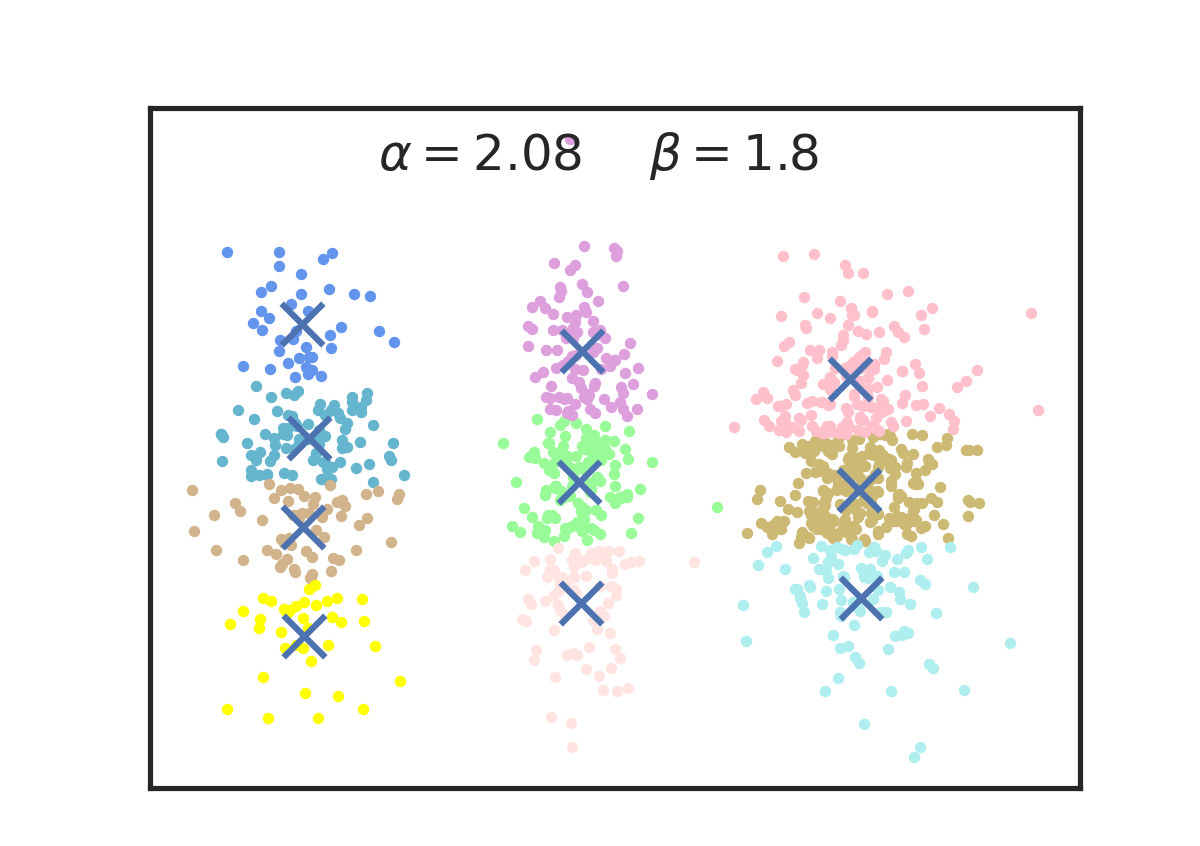}
         \caption{$k$-means++}
         \label{fig:three sin x}
     \end{subfigure}
     \begin{subfigure}[b]{0.33\linewidth}
         \centering
         \includegraphics[width=2.2in]{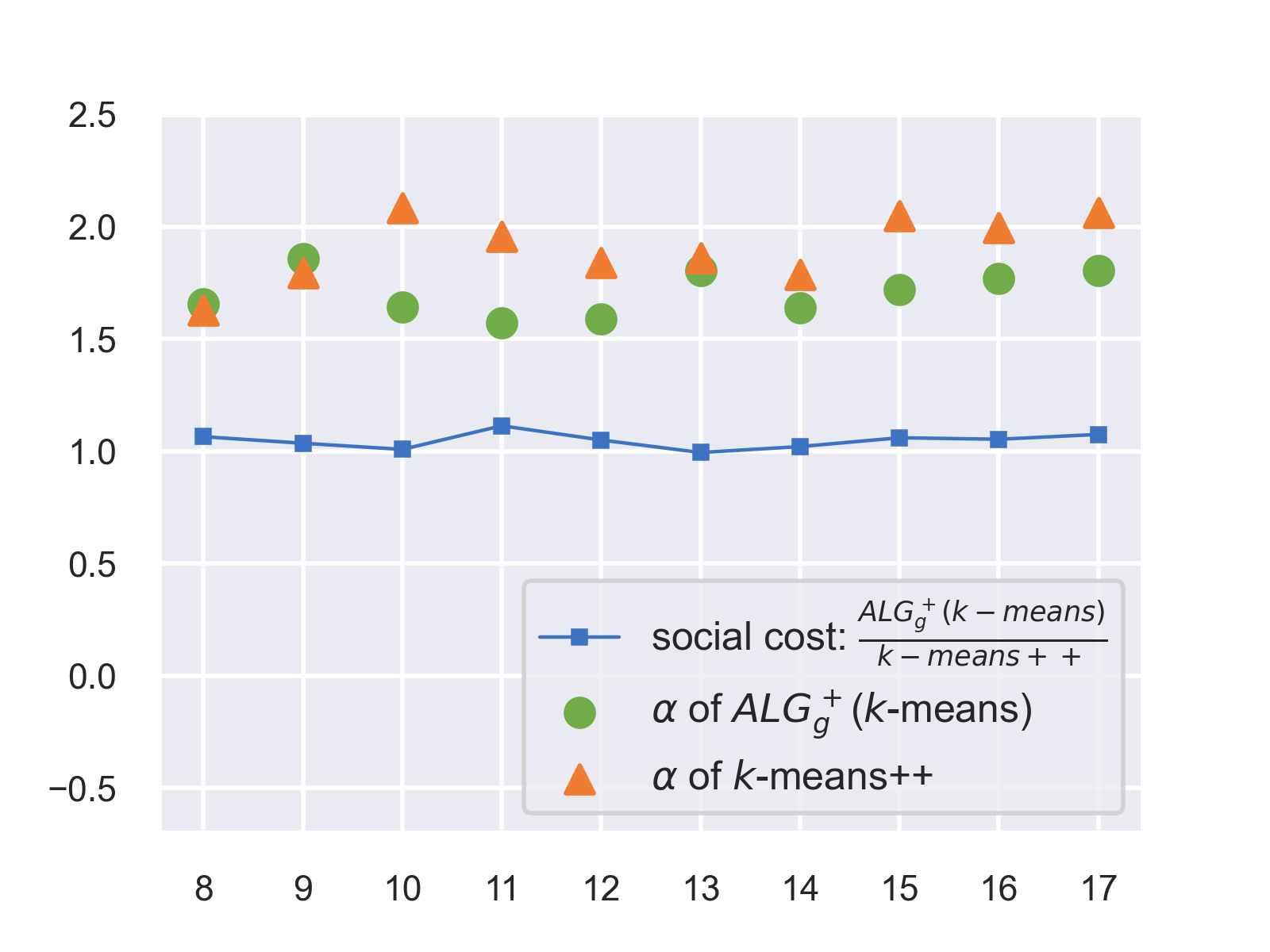}
         \caption{$\alpha$ and social cost in Gaussian dataset}
         \label{fig:zzz}
     \end{subfigure}
     
     \begin{subfigure}[b]{0.33\linewidth}
         \centering
         \includegraphics[width=2.2in]{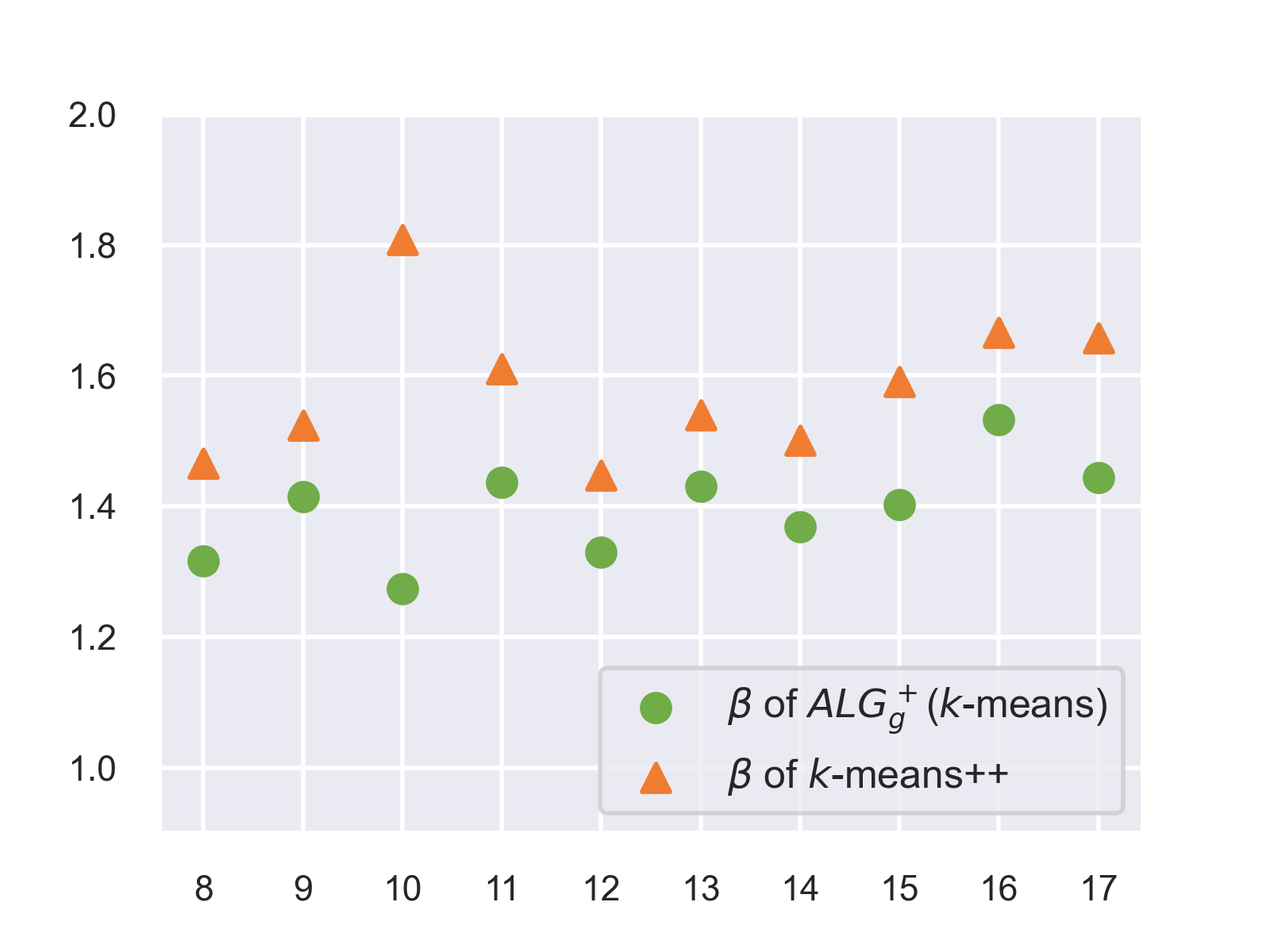}
         \caption{$\beta$ in Gaussian dataset}
         \label{fig:xxx}
     \end{subfigure}
     \begin{subfigure}[b]{0.33\linewidth}
         \centering
         \includegraphics[width=2.2in]{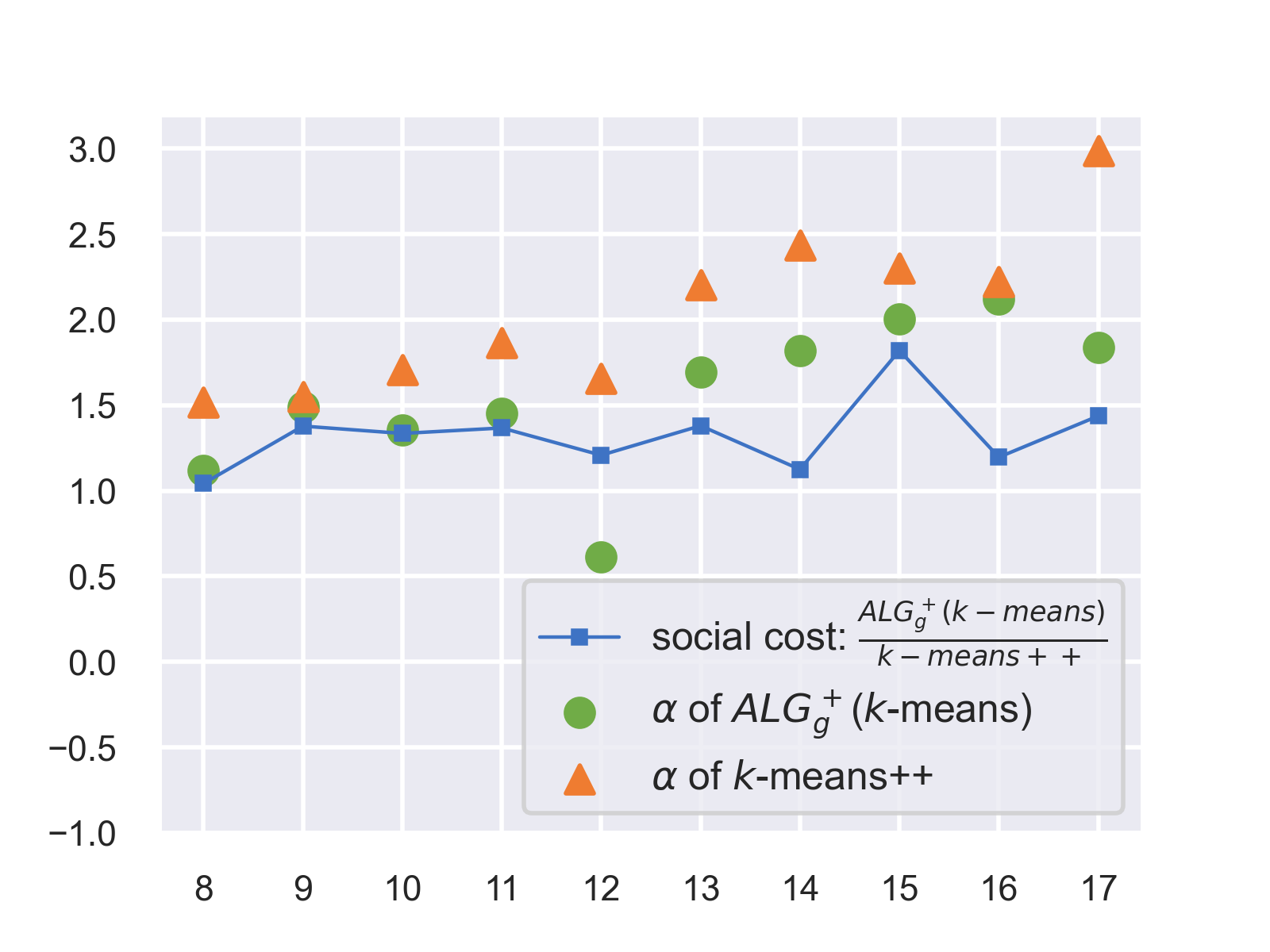}
         \caption{$\alpha$ and social cost in Mopsi locations}
         \label{fig:ccc}
     \end{subfigure}
     \begin{subfigure}[b]{0.33\linewidth}
         \centering
         \includegraphics[width=2.2in]{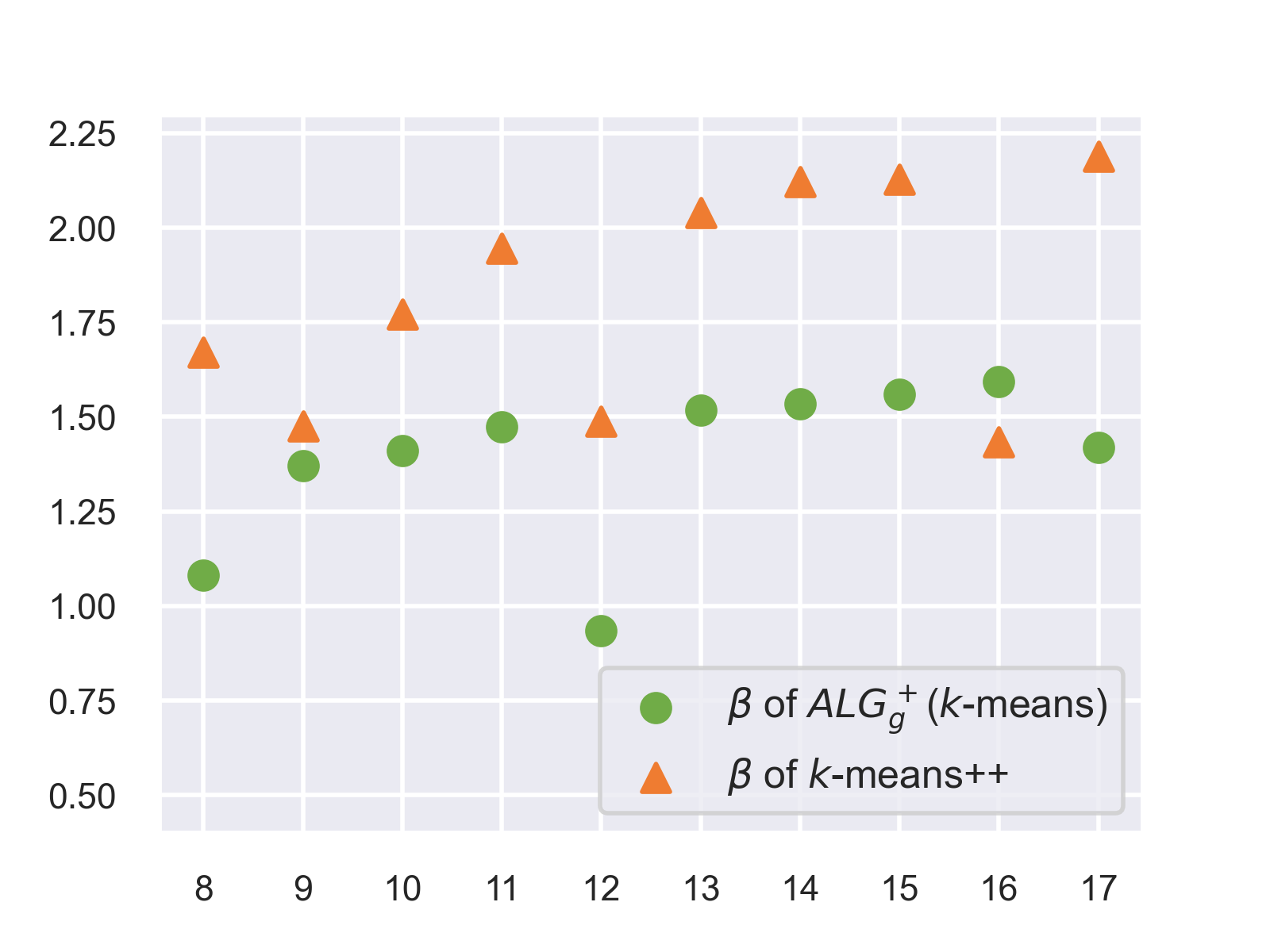}
         \caption{$\beta$ in Mopsi locations}
         \label{fig:vvv}
     \end{subfigure}
        \caption{(a) and (b) depict the clustering centers when two algorithms cluster Gaussian dataset for $k=10$. For a range of $k=8,\ldots,17$ (horizontal axis), (c) and (d) (resp. (e) and (f)) compare the fairness and efficiency in Gaussian dataset (resp. Mopsi locations). }
        \label{fig:bbb}
\end{figure}

\begin{algorithm}[H]
	\caption{\hspace{-2pt}{ \bf $ALG_g^+(\obj)$ for General Metric Space.}}
	\label{alg:heuristic}
	\begin{algorithmic}[1]
	\REQUIRE Metric space $(\mathcal{X},d)$, agents $\cN \subseteq \cX$, possible locations $\cM \subseteq \cX$, and $k \in \mathbb{N}^{+}$
	\ENSURE $k$-clustering $C$.
	
	\STATE Initialize $C = \emptyset$.
	
	\STATE Run $ALG_g$ on $(\cM,\cN,k)$ and get $Y=\{y_1,\ldots,y_{k'}\}$.
	\STATE Let $\cN_i$ be the corresponding cluster centered at $y_i$.
	\STATE 
	Rename so that $(|\cN_1| \mod \frac{n}{k})\ge\ldots\ge (|\cN_{k'}|\mod \frac{n}{k})$. 
	
	\STATE Let $r = k -\sum_{i=1}^{k'} \lfloor \frac{|\cN_i|}{n/k} \rfloor$.
	
	\FOR{$i = 1,\cdots,k'$}
	    \STATE Let $k_i = \lceil \frac{|\cN_i|}{n/k} \rceil$ if $i \le r$; otherwise, $k_i = \lfloor \frac{|\cN_i|}{n/k} \rfloor$.
	    \STATE \label{step:algg:obj} $\{y_{i1},\ldots,y_{ik_i}\} = \arg\min \obj(\cM, \cN_i, k_i)$. 
	    
	    \STATE $C\leftarrow C\cup \{y_{i1},\ldots,y_{ik_i}\}$.
	\ENDFOR
	\end{algorithmic}
\end{algorithm}

\vspace{-3mm}
To improve the performance, we refine $ALG_g$ in Algorithm \ref{alg:heuristic}, denoted by $ALG_g^+(\obj)$. Roughly, we first use $ALG_g$ to obtain a preliminary clustering {$Y$ and resultant partition $\cN = (\cN_i)_{i}$}, 
and then we proportionally assign all centers to these clusters according to their populations, i.e., $\sum_i k_i = k$ and $k_i \propto |\cN_i|$. 
Within each preliminary cluster $\cN_i$, the real centers are selected to optimize a social objective function $\obj(\cM,\cN_i,k_i)$ in Line 8. For example, when $\obj$ is the $k$-means objective (i.e., minimizing the squared Euclidean distances), we refer our algorithm to $ALG_g^+(k\text{-means})$; when $\obj$ is the $k$-medians objective (i.e., minimizing the Manhattan distances), we refer our algorithm to $ALG_g^+(k\text{-medians})$.
Thus, $\obj$ actually provides us an interface to balance fairness and social efficiency.
In the experiment shown in Figure \ref{fig:y equals x}, we feed $ALG_g^+(k\text{-means})$ with $k$-means++ algorithm \cite{vassilvitskii2006k}\footnote{$k$-means++ algorithm is Lloyd’s algorithm for k-means minimization objective with a particular initialization.},
which builds centers proportionally to the populations of the three Gaussian sets (i.e., 2,3,5 centers for each).
However, if we directly use $k$-means++ algorithm on $\cN$, it builds 4,3,3 centers for each Gaussian cluster, as shown in Figure \ref{fig:three sin x}, where the right set contains  half of all points but only gets 3 centers. 

\vspace{-2mm}\subsection{Experiments}
\vspace{-1mm}
We implement experiments on two qualitatively different datasets used for clustering. 
(1) Gaussian dataset (synthetic): the experiment in Section \ref{sec:000} (Figure \ref{fig:y equals x} and \ref{fig:three sin x}) is repeated for $k$ from 8 to 17.
(2) Mopsi locations in clustering benchmark datasets  \cite{ClusteringDatasets} (real-world): a set of 2-D locations for $n=6014$ users in Joensuu. 
Note that the second dataset concerning human beings is exactly the situation when the data points need to be treated fairly. Both
datasets are in Euclidean plane.

We consider the $k$-means objective as social cost, and compare our algorithm $ALG_g^+(k\text{-means})$ with $k$-means++. For each dataset, we consider a range of values of $k$. Figure \ref{fig:zzz}, \ref{fig:ccc} show the values of $\alpha$ (which is the minimum value such that the output clustering is a $(\alpha,1)$-core), and the ratio between the social costs of the two algorithms. Figure \ref{fig:xxx}, \ref{fig:vvv} show the values of $\beta$. 
In terms of fairness, $ALG_g^+(k\text{-means})$ has a significantly lower $\alpha$ and $\beta$ than $k$-means++ in most cases (though in few cases it is slightly larger). 
In terms of efficiency,  the social cost of $ALG_g^+(k\text{-means})$ is bounded by a small constant compared with $k$-means++, and ours is even as good as $k$-means++ on Gaussian dataset. 

Besides, we also investigate the $k$-medians objective empirically. We compared Lloyd's algorithm for $k$-medians and our Algorithm \ref{alg:heuristic} with $k$-medians objective in Line 8. The data sets consist of the 3-Gaussian dataset and real-world Mopsi locations in Joensuu, 
and S-sets in clustering benchmark datasets \cite{ClusteringDatasets}.

\textbf{Gaussian dataset.} We implement experiments with the $k$-medians objective for a range $k=8,\ldots,17$, as shown in Figure \ref{fig:kkkkk}. Our algorithm is clearly much fairer than Lloyd’s algorithm, in both $\alpha$-dimension and $\beta$-dimension. Moreover, for the social cost (i.e., the sum of the distances from each agent to the nearest center), our algorithm is slightly better than Lloyd’s algorithm, because our algorithm is two-stage which takes both fairness and efficiency into consideration. 

\begin{figure}[h]
     \centering
     \begin{subfigure}[h]{0.45\textwidth}
         \centering
         \includegraphics[width=0.83\textwidth]{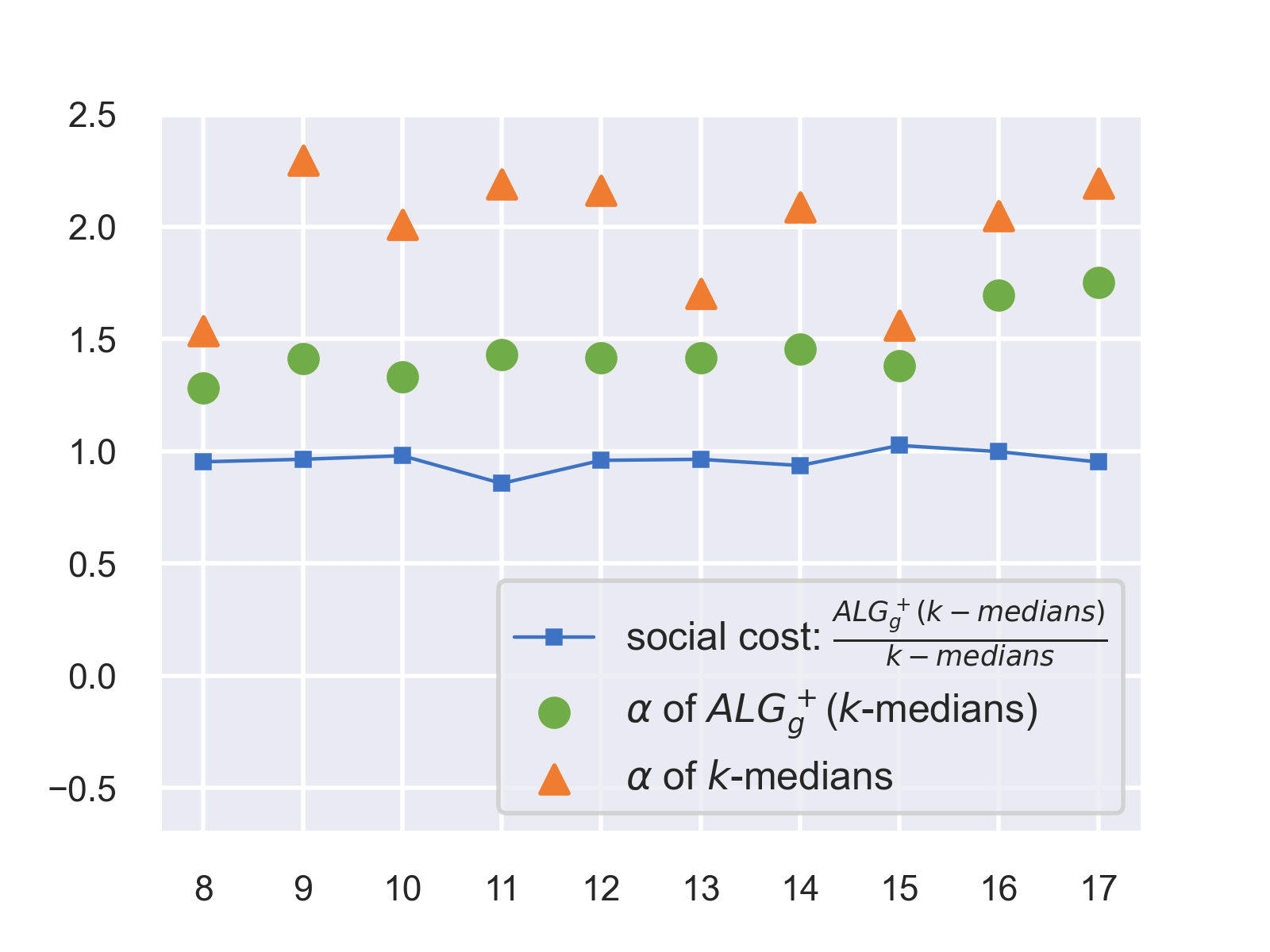}
         \caption{$\alpha$ and social cost on Gaussian dataset.}
     \end{subfigure}
     \begin{subfigure}[h]{0.45\textwidth}
         \centering
         \includegraphics[width=0.83\textwidth]{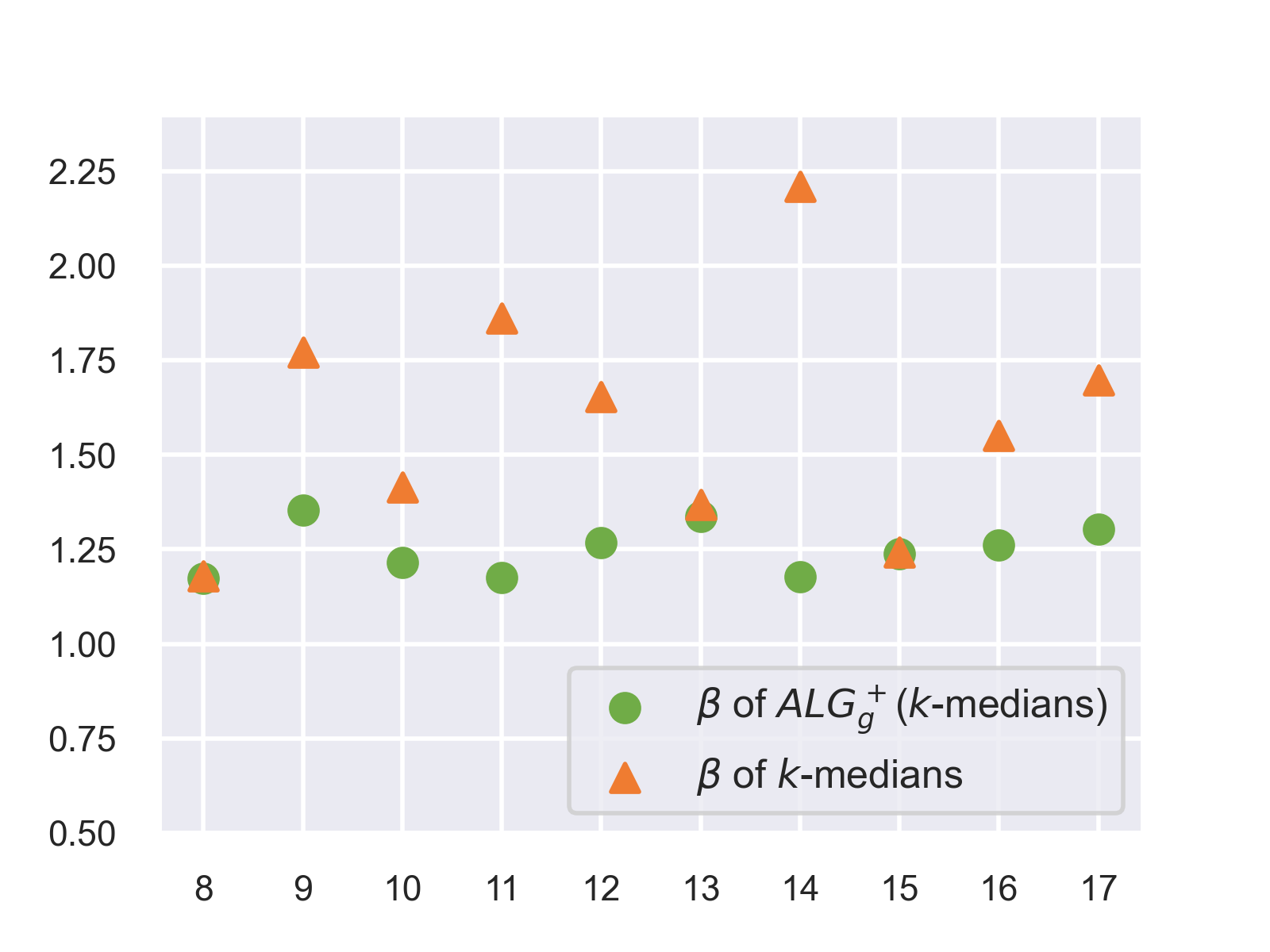}
         \caption{$\beta$ on Gaussian dataset.}
     \end{subfigure}
     \caption{The comparison between Algorithm \ref{alg:heuristic}  and Lloyd’s algorithm with $k$-medians objective on Gaussian dataset. }
     \label{fig:kkkkk}
\end{figure}

\begin{figure}[h]
     \centering
    \begin{subfigure}[b]{0.45\textwidth}
         \centering
         \includegraphics[width=0.65\textwidth]{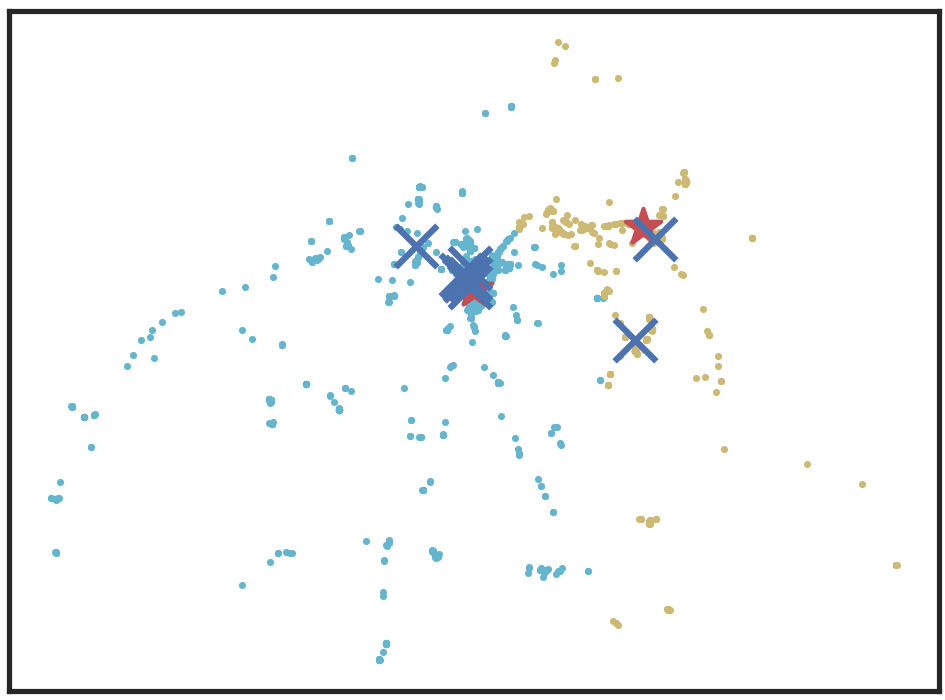}
         \caption{Algorithm \ref{alg:heuristic}.}
         \label{www}
     \end{subfigure}
     \begin{subfigure}[b]{0.45\textwidth}
         \centering
         \includegraphics[width=0.65\textwidth]{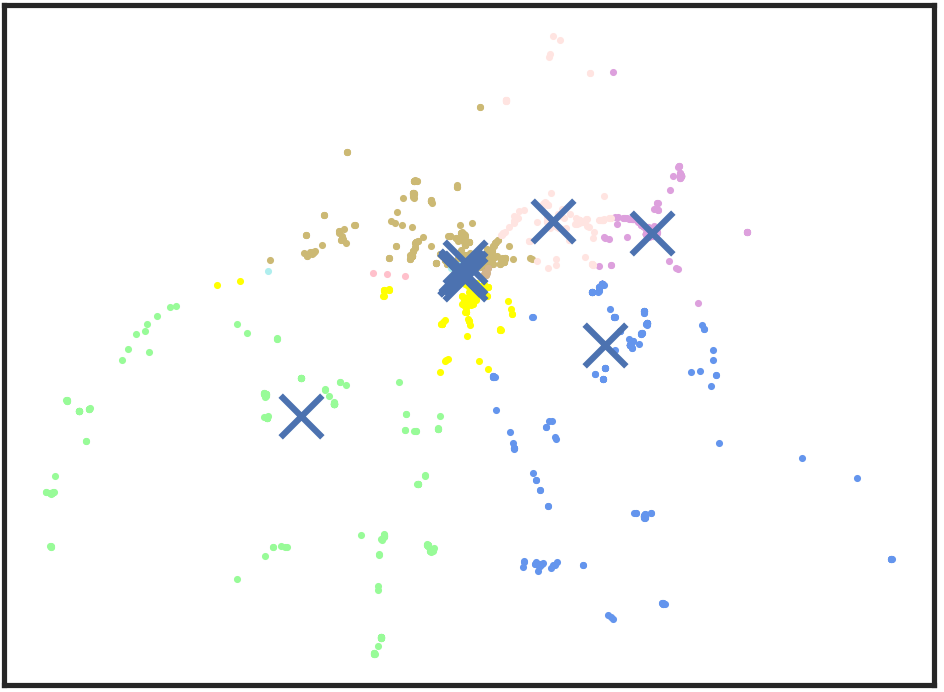}
         \caption{Lloyd’s algorithm.}
         \label{eee}
     \end{subfigure}
        \caption{Experiments with $k$-medians objective on Mopsi locations in Joensuu ($k=10$).}
        \label{fig:4321}
        
\end{figure}

\begin{figure}[h]
     \centering
     \begin{subfigure}[b]{0.45\textwidth}
         \centering
         \includegraphics[width=0.65\textwidth]{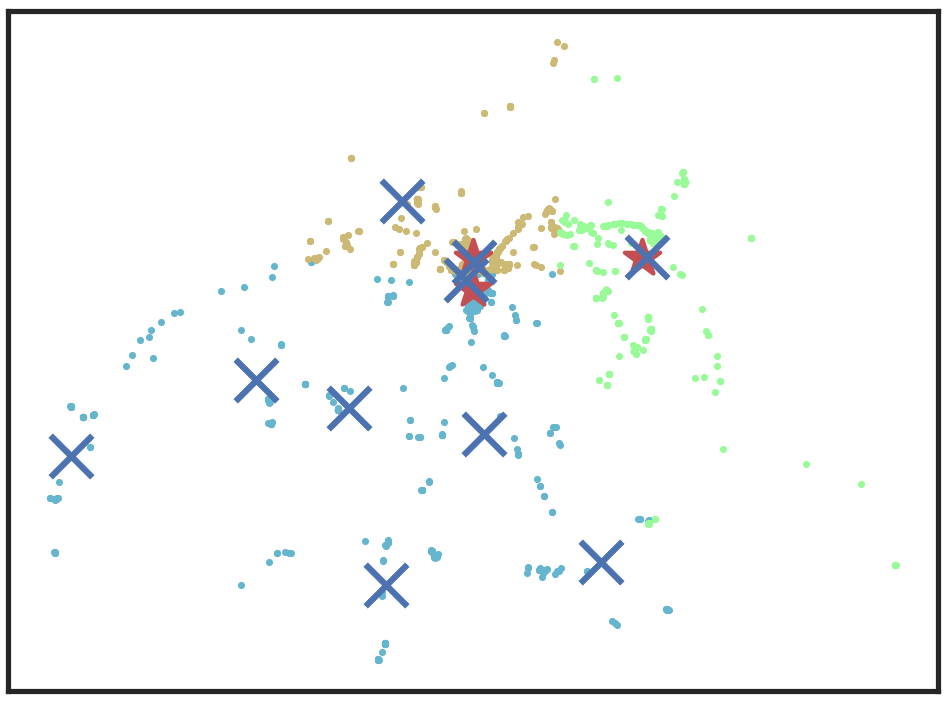}
         \caption{Algorithm \ref{alg:heuristic}.}
         \label{rrr}
     \end{subfigure}
     \begin{subfigure}[b]{0.45\textwidth}
         \centering
         \includegraphics[width=0.65\textwidth]{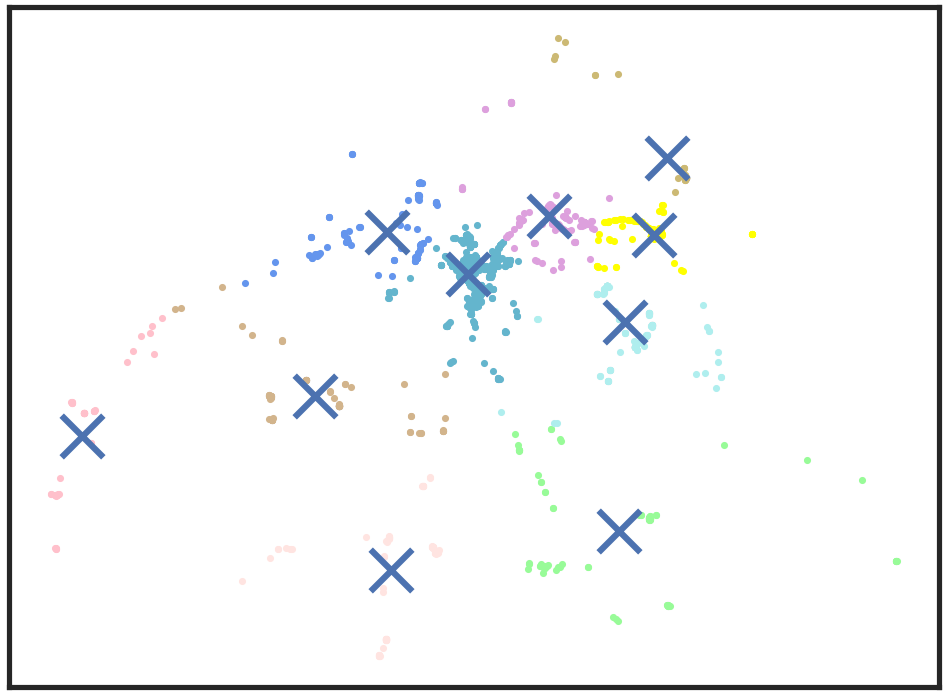}
         \caption{$k$-means++.}
         \label{ttt}
     \end{subfigure}
        \caption{Experiments with $k$-means objective on Mopsi locations in Joensuu ($k=10$).}
        \label{fig:loc}
        
\end{figure}

\textbf{Mopsi locations.} For the dataset of  Mopsi locations in Joensuu, we consider a range of $k=8,\ldots,17$. For every such $k$, the clustering output by either algorithm is exactly fair, i.e., $\alpha= 1,\beta= 1$. For example, we show the results with $k=10$ in Figure \ref{fig:4321}, where Figure \ref{www} depicts the clustering of Algorithm \ref{alg:heuristic}, and Figure \ref{eee} depicts the clustering of Lloyd’s algorithm. It is easy to observe that, in both clusterings, at least 6 centers are built nearly to serve a large group of agents, which helps guarantee the core fairness. The social cost ratio in this case is 1.38.

However, we notice that the algorithms with the $k$-means objective cannot obtain the exact core fairness, as shown in Figure \ref{fig:loc}. The clusterings returned by Algorithm \ref{alg:heuristic} is given by Figure \ref{rrr}, with $\alpha=1.49,\beta=1.45$.
The clusterings returned by $k$-means++ are given by Figure \ref{ttt}, with $\alpha=1.65,\beta=1.67$. The social cost ratio is 1.42.

\vspace{3mm}



\textbf{S-sets.} It contains synthetic 2-d data with $n=5000$ vectors and 15 Gaussian clusters with different degree of cluster overlap. Figure \ref{fig:s1a} and \ref{fig:s1b} show the experiment results on S1 set with the $k$-medians and $k$-means objectives, respectively. It can be seen that, for both objectives, our algorithm is fairer in most cases than classic algorithms, and the social cost of our algorithm is also better. 

To conclude, our algorithm ensures better core fairness for the agents than classic clustering algorithms, and meanwhile, empirically has a good efficiency as well.

\begin{figure}[h]
     \centering     
     \begin{subfigure}[h]{0.45\textwidth}
         \centering
         \includegraphics[width=0.83\textwidth]{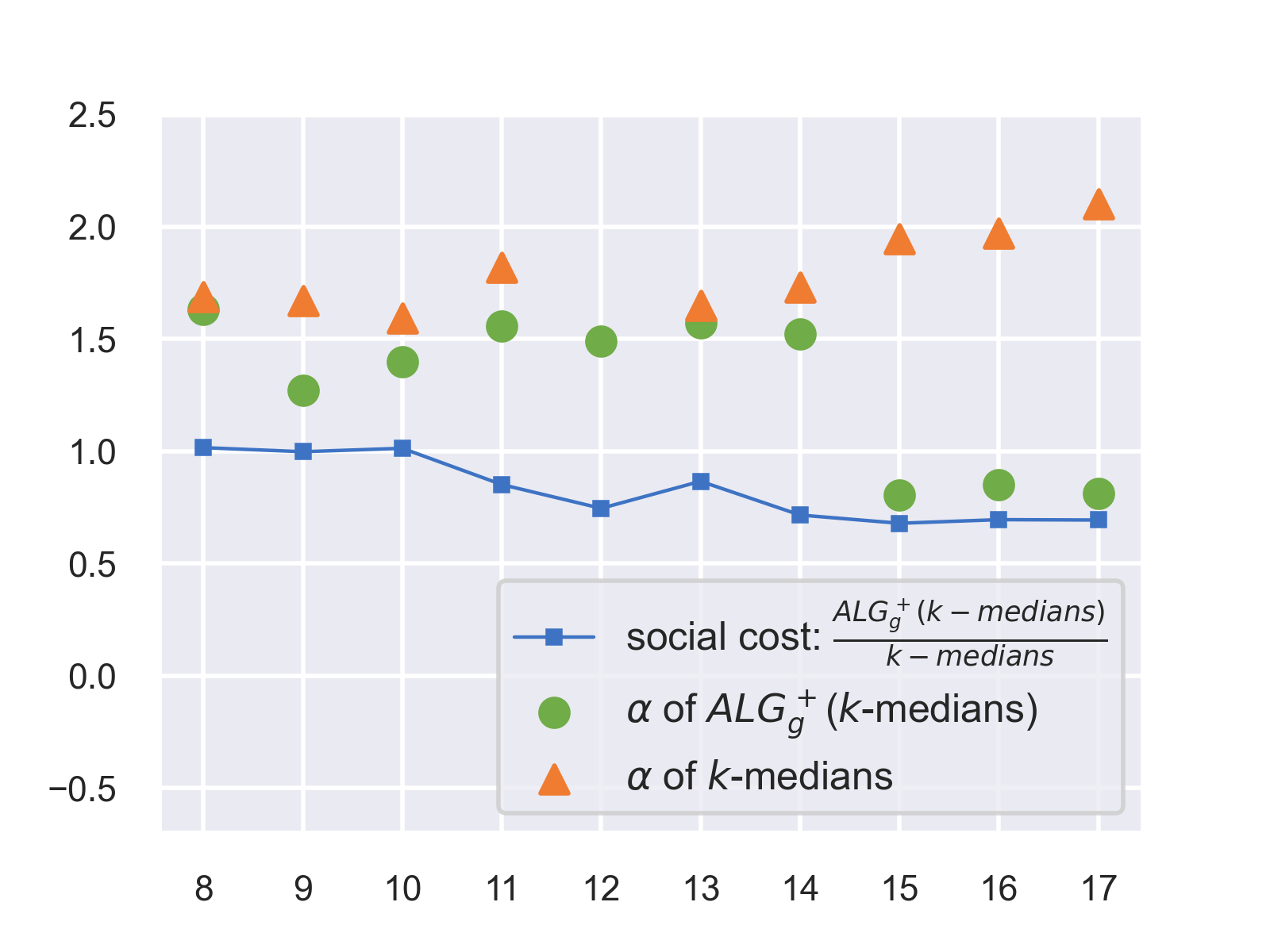}
         \caption{$\alpha$ and social cost on S1 dataset.}
     \end{subfigure}
     \begin{subfigure}[h]{0.45\textwidth}
         \centering
         \includegraphics[width=0.83\textwidth]{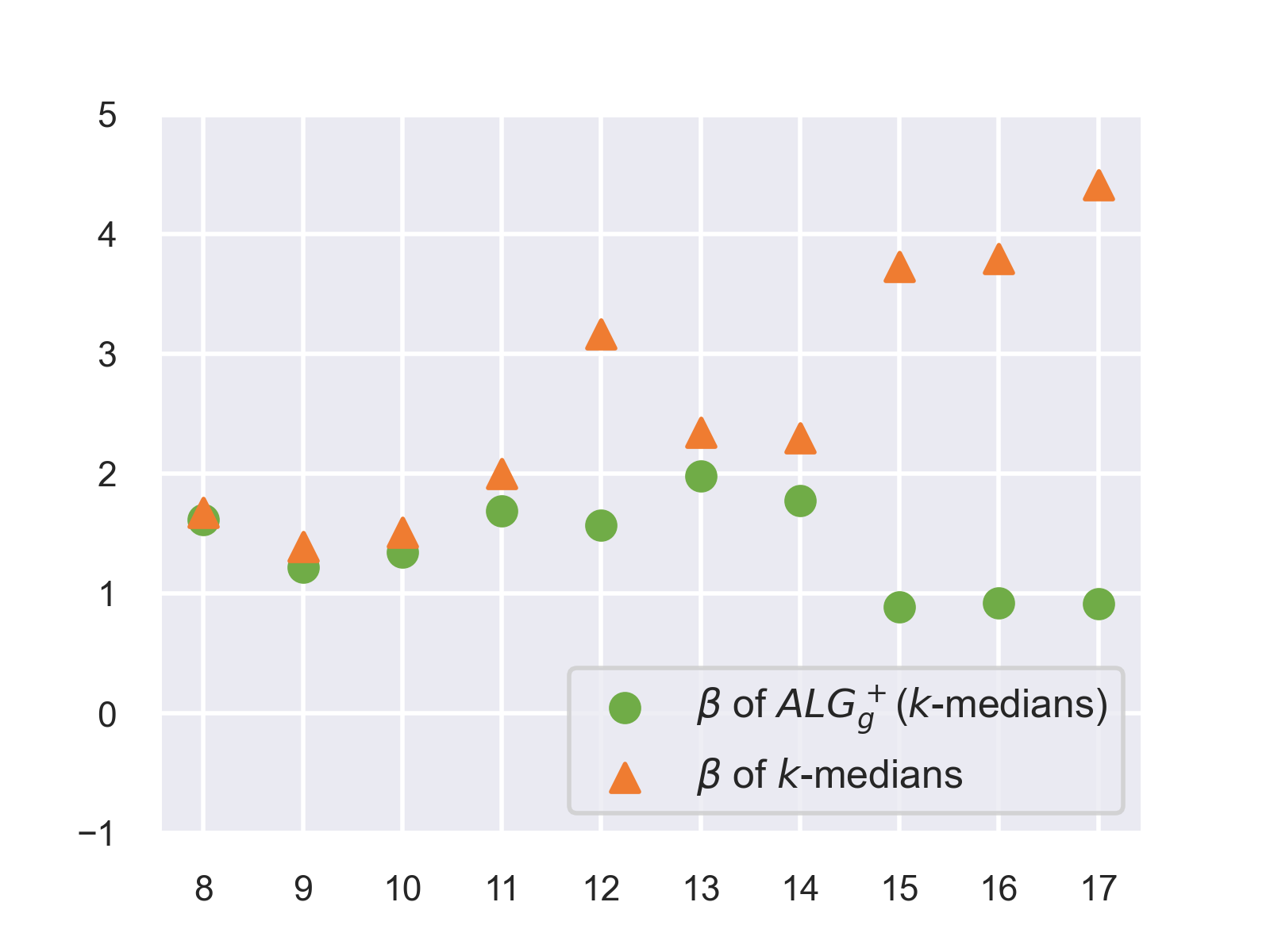}
         \caption{$\beta$ on S1 dataset.}
     \end{subfigure}
     
    \caption{Experiments with $k$-medians objective on S1 set.}
     \label{fig:s1a}
\end{figure}

\begin{figure}[h]
     \centering
     \begin{subfigure}[h]{0.45\textwidth}
         \centering
         \includegraphics[width=0.83\textwidth]{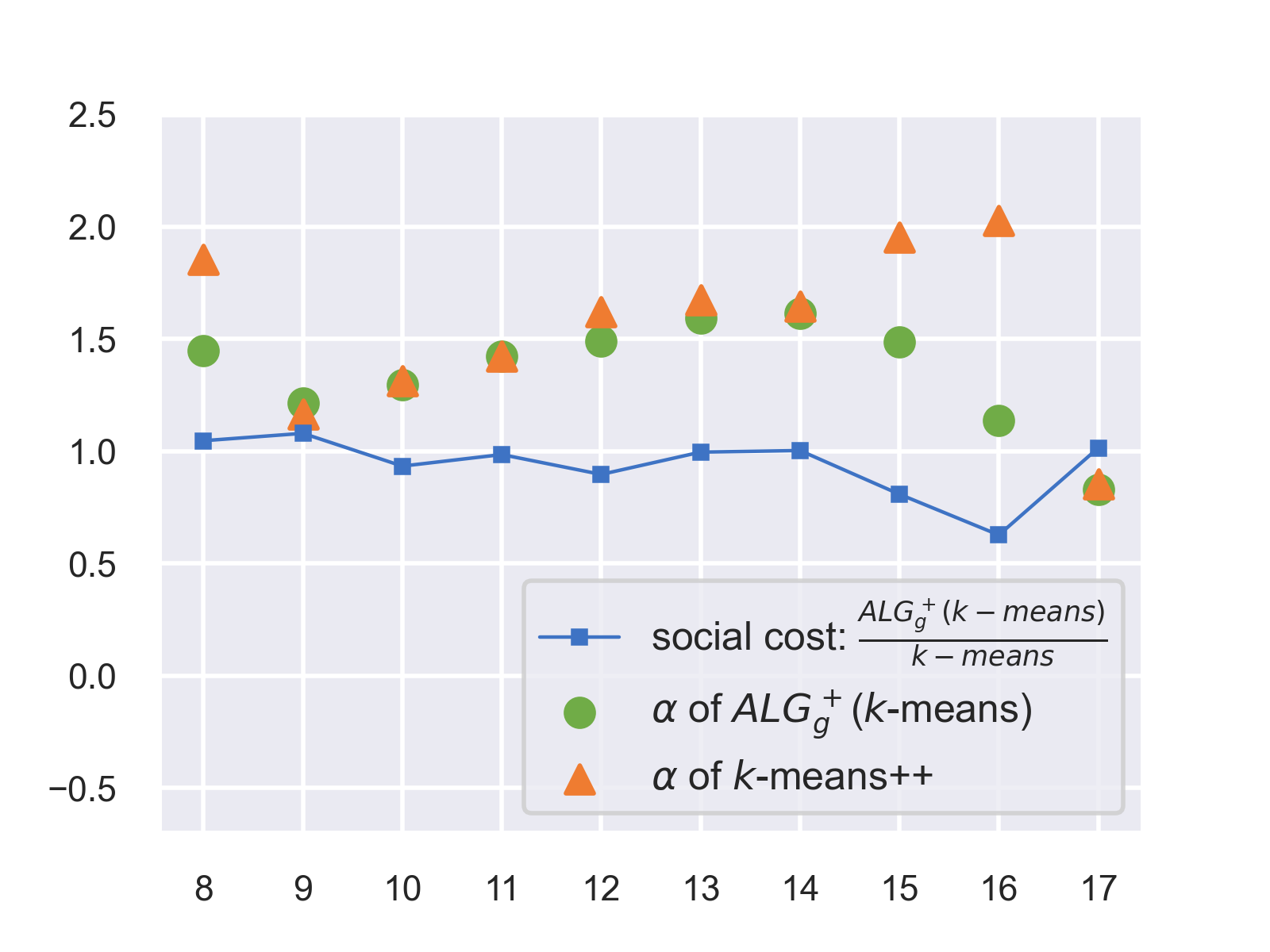}
         \caption{$\alpha$ and social cost on S1 dataset.}
     \end{subfigure}
     \begin{subfigure}[h]{0.45\textwidth}
         \centering
         \includegraphics[width=0.83\textwidth]{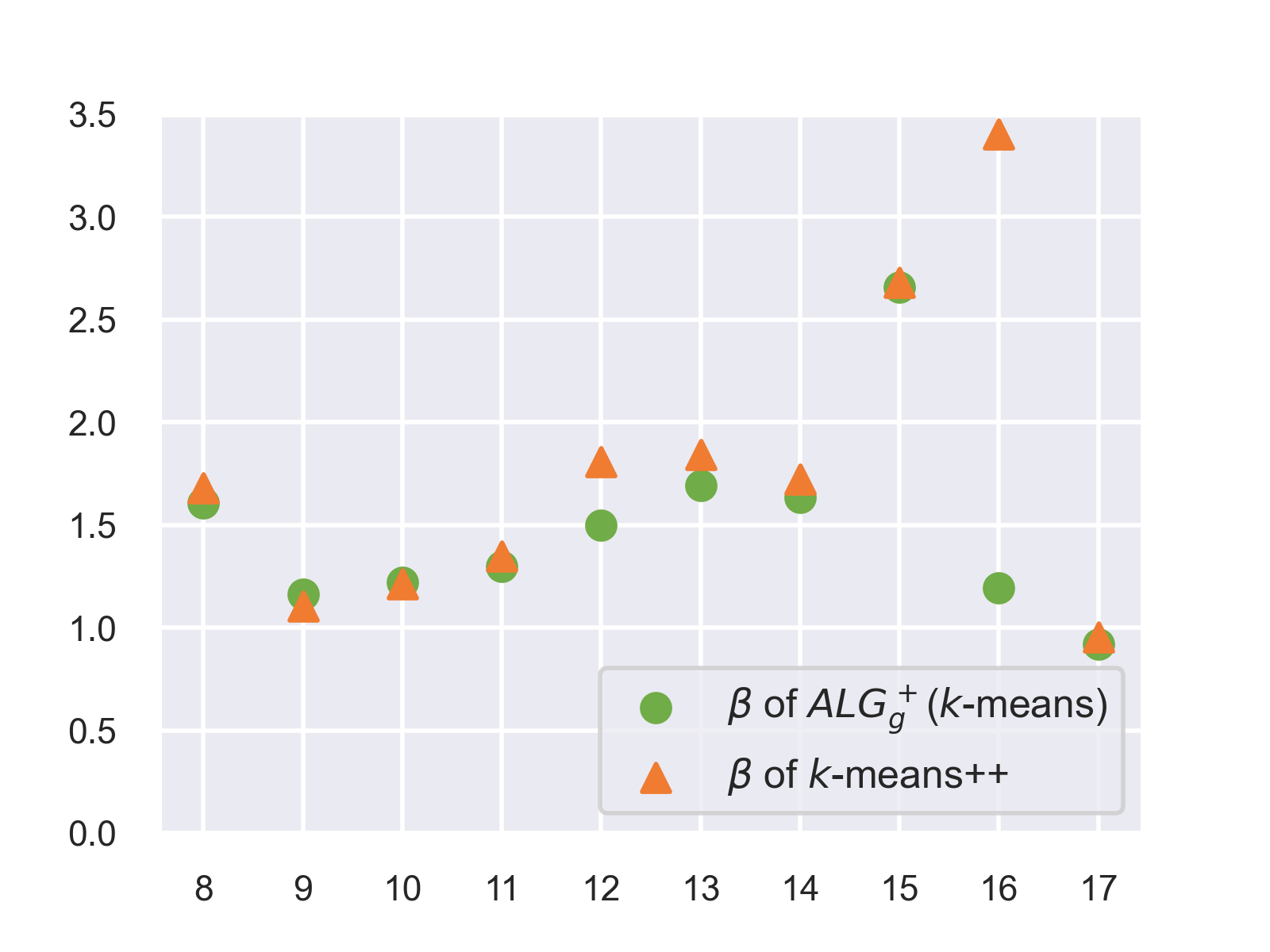}
         \caption{$\beta$ on S1 dataset.}
     \end{subfigure}
        \caption{Experiments with $k$-means objective on S1 set.}
        \label{fig:s1b}
\end{figure}

\section{Conclusion}
In this work, we revisited the algorithmic fair centroid clustering problem, and proposed a novel definition for group fairness -- core.
We demonstrated the extent to which an approximate core clustering is guaranteed to exist.

There are many future directions that are worth exploring.
For example, it will be interesting to improve the approximation bounds for other spaces such as $d$-dimensional Euclidean space for $d\ge 2$, and simple graphs with unit weight.
In Appendix \ref{appendix:B}, we prove that even when all agents consist of all vertices of a unit-weight tree, an exact core can still be empty. A systematic analysis of upper- and lower-bounds is left for future study.

\section*{Acknowledgments}
This work is partially supported by Research Grants Council of HKSAR under Grant No. PolyU 25211321, NSFC under Grant No. 62102333, and The Hong Kong Polytechnic University under Grant No. P0034420.

\bibliographystyle{unsrt}
\bibliography{main}

\begin{thebibliography}{10}

\bibitem{chierichetti2017fair}
Flavio Chierichetti, Ravi Kumar, Silvio Lattanzi, and Sergei Vassilvitskii.
\newblock Fair clustering through fairlets.
\newblock In {\em {NIPS}}, pages 5029--5037, 2017.

\bibitem{chen2019proportionally}
Xingyu Chen, Brandon Fain, Liang Lyu, and Kamesh Munagala.
\newblock Proportionally fair clustering.
\newblock In {\em {ICML}}, volume~97 of {\em Proceedings of Machine Learning
  Research}, pages 1032--1041. {PMLR}, 2019.

\bibitem{backurs2019scalable}
Arturs Backurs, Piotr Indyk, Krzysztof Onak, Baruch Schieber, Ali Vakilian, and
  Tal Wagner.
\newblock Scalable fair clustering.
\newblock In {\em {ICML}}, volume~97 of {\em Proceedings of Machine Learning
  Research}, pages 405--413. {PMLR}, 2019.

\bibitem{bera2019fair}
Suman~Kalyan Bera, Deeparnab Chakrabarty, Nicolas Flores, and Maryam Negahbani.
\newblock Fair algorithms for clustering.
\newblock In {\em NeurIPS}, pages 4955--4966, 2019.

\bibitem{conitzer2017fair}
Vincent Conitzer, Rupert Freeman, and Nisarg Shah.
\newblock Fair public decision making.
\newblock In {\em Proceedings of the 18th ACM Conference on Economics and
  Computation (EC)}, pages 629--646, 2017.

\bibitem{fain2018fair}
Brandon Fain, Kamesh Munagala, and Nisarg Shah.
\newblock Fair allocation of indivisible public goods.
\newblock In {\em {EC}}, pages 575--592. {ACM}, 2018.

\bibitem{michaproportionally}
Evi Micha and Nisarg Shah.
\newblock Proportionally fair clustering revisited.
\newblock In {\em {ICALP}}, volume 168 of {\em LIPIcs}, pages 85:1--85:16.
  Schloss Dagstuhl - Leibniz-Zentrum f{\"{u}}r Informatik, 2020.

\bibitem{moulin2004fair}
Herv{\'{e}} Moulin.
\newblock {\em Fair division and collective welfare}.
\newblock {MIT} Press, 2003.

\bibitem{deng1994complexity}
Xiaotie Deng and Christos~H. Papadimitriou.
\newblock On the complexity of cooperative solution concepts.
\newblock {\em Math. Oper. Res.}, 19(2):257--266, 1994.

\bibitem{chalkiadakis2011computational}
Georgios Chalkiadakis, Edith Elkind, and Michael~J. Wooldridge.
\newblock {\em Computational Aspects of Cooperative Game Theory}.
\newblock Synthesis Lectures on Artificial Intelligence and Machine Learning.
  Morgan {\&} Claypool Publishers, 2011.

\bibitem{mckelvey1986covering}
Richard~D McKelvey.
\newblock Covering, dominance, and institution-free properties of social
  choice.
\newblock {\em American Journal of Political Science}, pages 283--314, 1986.

\bibitem{feldman2006welfare}
Allan~M Feldman and Roberto Serrano.
\newblock {\em Welfare Economics and Social Choice Theory}.
\newblock Springer Science \& Business Media, 2006.

\bibitem{abdulkadirouglu1998random}
Atila Abdulkadiro{\u{g}}lu and Tayfun S{\"o}nmez.
\newblock Random serial dictatorship and the core from random endowments in
  house allocation problems.
\newblock {\em Econometrica}, 66(3):689--701, 1998.

\bibitem{feldman2013strategyproof}
Michal Feldman and Yoav Wilf.
\newblock Strategyproof facility location and the least squares objective.
\newblock In {\em {EC}}, pages 873--890. {ACM}, 2013.

\bibitem{kleinberg2005query}
Jon~M. Kleinberg and Prabhakar Raghavan.
\newblock Query incentive networks.
\newblock In {\em {FOCS}}, pages 132--141. {IEEE} Computer Society, 2005.

\bibitem{babaioff2012bitcoin}
Moshe Babaioff, Shahar Dobzinski, Sigal Oren, and Aviv Zohar.
\newblock On bitcoin and red balloons.
\newblock In {\em {EC}}, pages 56--73. {ACM}, 2012.

\bibitem{alon2010strategyproof}
Noga Alon, Michal Feldman, Ariel~D. Procaccia, and Moshe Tennenholtz.
\newblock Strategyproof approximation of the minimax on networks.
\newblock {\em Math. Oper. Res.}, 35(3):513--526, 2010.

\bibitem{DBLP:journals/mor/Myerson77}
Roger~B. Myerson.
\newblock Graphs and cooperation in games.
\newblock {\em Math. Oper. Res.}, 2(3):225--229, 1977.

\bibitem{DBLP:conf/aaai/MeirZER13}
Reshef Meir, Yair Zick, Edith Elkind, and Jeffrey~S. Rosenschein.
\newblock Bounding the cost of stability in games over interaction networks.
\newblock In {\em {AAAI}}. {AAAI} Press, 2013.

\bibitem{DBLP:conf/wine/Elkind14}
Edith Elkind.
\newblock Coalitional games on sparse social networks.
\newblock In {\em {WINE}}, volume 8877 of {\em Lecture Notes in Computer
  Science}, pages 308--321. Springer, 2014.

\bibitem{DBLP:conf/sagt/BachrachEMPZRR09}
Yoram Bachrach, Edith Elkind, Reshef Meir, Dmitrii~V. Pasechnik, Michael
  Zuckerman, J{\"{o}}rg Rothe, and Jeffrey~S. Rosenschein.
\newblock The cost of stability in coalitional games.
\newblock In {\em {SAGT}}, volume 5814 of {\em Lecture Notes in Computer
  Science}, pages 122--134. Springer, 2009.

\bibitem{DBLP:journals/mor/MaschlerPS79}
M.~Maschler, B.~Peleg, and L.~S. Shapley.
\newblock Geometric properties of the kernel, nucleolus, and related solution
  concepts.
\newblock {\em Math. Oper. Res.}, 4(4):303--338, 1979.

\bibitem{harb2020kfc}
Elfarouk Harb and Ho~Shan Lam.
\newblock {KFC:} {A} scalable approximation algorithm for \emph{k}-center fair
  clustering.
\newblock In {\em NeurIPS}, 2020.

\bibitem{bercea2019cost}
Ioana~Oriana Bercea, Martin Gro{\ss}, Samir Khuller, Aounon Kumar, Clemens
  R{\"{o}}sner, Daniel~R. Schmidt, and Melanie Schmidt.
\newblock On the cost of essentially fair clusterings.
\newblock In {\em {APPROX-RANDOM}}, volume 145 of {\em LIPIcs}, pages
  18:1--18:22. Schloss Dagstuhl - Leibniz-Zentrum f{\"{u}}r Informatik, 2019.

\bibitem{huang2019coresets}
Lingxiao Huang, Shaofeng~H.{-}C. Jiang, and Nisheeth~K. Vishnoi.
\newblock Coresets for clustering with fairness constraints.
\newblock In {\em NeurIPS}, pages 7587--7598, 2019.

\bibitem{schmidt2018fair}
Melanie Schmidt, Chris Schwiegelshohn, and Christian Sohler.
\newblock Fair coresets and streaming algorithms for fair \emph{k}-means
  clustering.
\newblock {\em CoRR}, abs/1812.10854, 2018.

\bibitem{rosner2018privacy}
Clemens R{\"{o}}sner and Melanie Schmidt.
\newblock Privacy preserving clustering with constraints.
\newblock In {\em {ICALP}}, volume 107 of {\em LIPIcs}, pages 96:1--96:14.
  Schloss Dagstuhl - Leibniz-Zentrum f{\"{u}}r Informatik, 2018.

\bibitem{braverman2019coresets}
Vladimir Braverman, Shaofeng~H.{-}C. Jiang, Robert Krauthgamer, and Xuan Wu.
\newblock Coresets for ordered weighted clustering.
\newblock In {\em {ICML}}, volume~97 of {\em Proceedings of Machine Learning
  Research}, pages 744--753. {PMLR}, 2019.

\bibitem{DBLP:conf/nips/ZafarVGGW17}
Muhammad~Bilal Zafar, Isabel Valera, Manuel Gomez{-}Rodriguez, Krishna~P.
  Gummadi, and Adrian Weller.
\newblock From parity to preference-based notions of fairness in
  classification.
\newblock In {\em {NIPS}}, pages 229--239, 2017.

\bibitem{DBLP:conf/icml/UstunLP19}
Berk Ustun, Yang Liu, and David~C. Parkes.
\newblock Fairness without harm: decoupled classifiers with preference
  guarantees.
\newblock In {\em {ICML}}, volume~97 of {\em Proceedings of Machine Learning
  Research}, pages 6373--6382. {PMLR}, 2019.

\bibitem{DBLP:conf/nips/BalcanDNP19}
Maria{-}Florina Balcan, Travis Dick, Ritesh Noothigattu, and Ariel~D.
  Procaccia.
\newblock Envy-free classification.
\newblock In {\em NeurIPS}, pages 1238--1248, 2019.

\bibitem{mehrabi2019survey}
Ninareh Mehrabi, Fred Morstatter, Nripsuta Saxena, Kristina Lerman, and Aram
  Galstyan.
\newblock A survey on bias and fairness in machine learning.
\newblock {\em CoRR}, abs/1908.09635, 2019.

\bibitem{segal2019democratic}
Erel Segal{-}Halevi and Warut Suksompong.
\newblock Democratic fair allocation of indivisible goods.
\newblock {\em Artif. Intell.}, 277, 2019.

\bibitem{benabbou2019fairness}
Nawal Benabbou, Mithun Chakraborty, Edith Elkind, and Yair Zick.
\newblock Fairness towards groups of agents in the allocation of indivisible
  items.
\newblock In {\em {IJCAI}}, pages 95--101. ijcai.org, 2019.

\bibitem{kyropoulou2020almost}
Maria Kyropoulou, Warut Suksompong, and Alexandros~A. Voudouris.
\newblock Almost envy-freeness in group resource allocation.
\newblock {\em Theor. Comput. Sci.}, 841:110--123, 2020.

\bibitem{aziz2019almost}
Haris Aziz and Simon Rey.
\newblock Almost group envy-free allocation of indivisible goods and chores.
\newblock In {\em {IJCAI}}, pages 39--45. ijcai.org, 2020.

\bibitem{berliant1992fair}
Marcus Berliant, William Thomson, and Karl Dunz.
\newblock On the fair division of a heterogeneous commodity.
\newblock {\em Journal of Mathematical Economics}, 21(3):201--216, 1992.

\bibitem{conitzer2019group}
Vincent Conitzer, Rupert Freeman, Nisarg Shah, and Jennifer~Wortman Vaughan.
\newblock Group fairness for the allocation of indivisible goods.
\newblock In {\em {AAAI}}, pages 1853--1860. {AAAI} Press, 2019.

\bibitem{hossain2020designing}
Safwan Hossain, Andjela Mladenovic, and Nisarg Shah.
\newblock Designing fairly fair classifiers via economic fairness notions.
\newblock In {\em {WWW}}, pages 1559--1569. {ACM} / {IW3C2}, 2020.

\bibitem{cheng2019group}
Yu~Cheng, Zhihao Jiang, Kamesh Munagala, and Kangning Wang.
\newblock Group fairness in committee selection.
\newblock In {\em {EC}}, pages 263--279. {ACM}, 2019.

\bibitem{li2020fair}
Bo~Li and Yingkai Li.
\newblock Fair resource sharing and dorm assignment.
\newblock In {\em {AAMAS}}, pages 708--716. International Foundation for
  Autonomous Agents and Multiagent Systems, 2020.

\bibitem{jain1999primal}
Kamal Jain and Vijay~V. Vazirani.
\newblock Primal-dual approximation algorithms for metric facility location and
  \emph{k}-median problems.
\newblock In {\em {FOCS}}, pages 2--13. {IEEE} Computer Society, 1999.

\bibitem{vassilvitskii2006k}
David Arthur and Sergei Vassilvitskii.
\newblock k-means++: the advantages of careful seeding.
\newblock In {\em {SODA}}, pages 1027--1035. {SIAM}, 2007.

\bibitem{ClusteringDatasets}
Pasi Fr{\"{a}}nti and Sami Sieranoja.
\newblock K-means properties on six clustering benchmark datasets.
\newblock {\em Appl. Intell.}, 48(12):4743--4759, 2018.

\end{thebibliography}

\appendix
\begin{center}
    {\bf \Large Appendix}
\end{center}



\section{Classic Algorithms Can Be Arbitrarily Unfair}
\label{appendix:A}

Regarding core fairness, we note that the traditional learning algorithms (e.g., Lloyd-style algorithms for $k$-means and $k$-medians) can be arbitrarily unfair in the worst case. Take $k$-medians as an example. Recall that $k$-medians algorithm returns the cluster set $Y \in [\cM]^k$  minimizing the total distance.
Consider the following special setting when $\cX = \cM = \bR$ and $k = 1$. For this case, $k$-medians algorithm degenerates to {\em selecting the median point of $\cN$}. That is, given $n$ locations of agents with $a_1 \le a_2 \le \cdots \le a_n$, the algorithm outputs $Y =\{  a_{\lceil\frac{n}{2}\rceil} \}$. Similarly, when $k$ is arbitrary, $k$-medians algorithm partitions $\cN$ into $k$ subsets and within each subset, the median point of it is selected as the center.

It is not hard to verify that when $k = 1$, $k$-medians algorithm outputs a core clustering. However, when $k \ge 2$, $k$-medians algorithm can be arbitrarily unfair.
Consider $k$ groups (with different sizes) of points with a total size $n$, where each group is far from the others. Then $k$-medians algorithm will set one cluster center for each group (at its median point) to minimize the total distance. Assume there is a group of $2\left \lceil \frac{n}{k} \right \rceil+1$ points, where
$\lceil \frac{n}{k} \rceil$ of these points are located at $0$, a single point is at $1$, and the remaining $\lceil \frac{n}{k} \rceil$ are at $2$.
Accordingly, $k$-medians algorithm puts the cluster center for this group of points at $1$ to minimize the total $1$-norm distance. However, all the $\lceil \frac{n}{k} \rceil$ points located at $0$ can deviate to a new cluster center at $0$ that decreases their total distance to $0$; and thus
\[
\beta = \frac{\left \lceil \frac{n}{k} \right \rceil}{0} = \infty,
\]
which means $k$-medians algorithm returns an arbitrarily unfair clustering in the $\beta$-dimension relaxation.


\section{Lower Bound for Simple Trees}
\label{appendix:B}



\begin{theorem}
There exists an instance in an unweighted tree such that $(1,\beta)$-core and  $(\alpha,1)$-core for all $\beta < \frac{14}{13}$ and $\alpha<\frac{28}{25}$ are both empty, even if each vertex lies exactly one agent.
\end{theorem}

\begin{figure}[H]
    \centering
    \includegraphics[width=6cm]{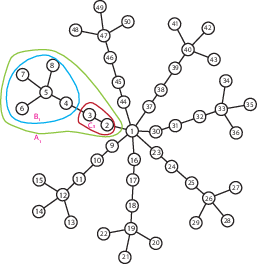}
    \caption{A lower bound for simple trees}
    \label{fig:unit-tree:lb}
\end{figure}

\begin{proof}
Consider the following instance constructed on the tree with unit weights. As shown in Figure \ref{fig:unit-tree:lb}, there are $n=50$ vertices distributed on $7$ branches and at each vertex lies an agent. Let $k=7$ be the number of cluster centers.

We first prove that $(1,\beta)$-core is empty for all $\beta < \frac{14}{13}$.
For clarity, in the figure we number the vertices and group them into $3$ types (sets): $A$, $B$ and $C$, each of which contains $7$ parts of vertices with identical structure on each branch denoted by $A=\bigcup_{i=1}^{7}A_i$, $B=\bigcup_{i=1}^{7}B_i$ and $C=\bigcup_{i=1}^{7}C_i$. Regarding deviation, by Lemma \ref{lem:nor}, we only need to consider a potential blocking coalition with size $\lceil \frac{n}{k} \rceil = 8$.

\emph{Case 1}: There exists at least one part of vertices in A without a center located within. W.l.o.g., suppose vertices in $A_1$ have no center located on them. Considering the group of agents $S=\{v_1,v_2,…,v_8\}$, their total distance from the center assigned to them in this case is at least $\sum_{i \in S} d(i,v_1) = 25$. However, they can form a blocking coalition to deviate together to a new center $v_5$, which can decrease their total distance to the lowest(optimal): $\sum_{i \in S} d(i,v_5) = 13$. Therefore, for $\beta < \frac{25}{13}$, all the clustering solutions in Case 1 are not a $(1,\beta)$-core clustering.

\emph{Case 2}: There is one center in every part of A type vertices. We further discuss two subcases.

\emph{Subcase 2.1}: There is exact one center in every part of B type vertices.
Considering a potential blocking coalition of $8$ agents $S=\{ v_1,v_2,v_9,v_{16},v_{23},v_{30},v_{37},v_{44} \}$ located at the center part of this tree, their total distance from the center assigned to them is at least $d(v_2,v_4)\cdot 7+d(v_1,v_4) = 17$. They can deviate to a new center $v_1$, which can decrease their total distance to $\sum_{i \in S} d(i,v_1) = 7$. Therefore, for $\beta < \frac{17}{7}$, all the clustering solutions in Subcase 2.1 are not a $(1,\beta)$-core clustering.

\emph{Subcase 2.2}: There exists one part of B without centers located in it, that is, there exists one part of C having a center located in it. W.l.o.g., suppose $C_1$ has a center located within while $B_1$ does not have one.
Consider the group of agents $S=\{v_1,v_2,…,v_8\}$. Their total distance from the centers assigned to them is at least $\sum_{i \in S} \min \{d(i,v_3),d(i,v_9\}) = 14$. However, they can form a blocking coalition to deviate together to a new center $v_5$, which can decrease their total distance to $\sum_{i \in S} d(i,v_5) = 13$. Hence, for $\beta < \frac{14}{13}$, all the clustering solutions in Subcase 2.2 are not a $(1,\beta)$-core clustering.

Therefore, in this instance, for all $\beta < \frac{14}{13}$, $(1,\beta)$-core is empty.

When $\beta=1$ and $\alpha<\frac{56}{50}=\frac{28}{25}$, for the above instance there always exists a blocking coalition of size $\lceil\frac{\alpha n}{k}\rceil=8$. Therefore, $(\alpha,1)$-core is empty for all $\alpha<\frac{28}{25}$.
\end{proof}




\end{document}